\documentclass[aps,superscriptaddress,showpacs,pra,10pt,nofootinbib,notitlepage]{revtex4-1}
\usepackage{comment}   
\usepackage{dsfont}   
\usepackage{xcolor} 
\usepackage{amsmath,amsthm,amssymb}  
 \newtheorem{thm}{Theorem}
 \newtheorem{lemma}{Lemma}
 \newtheorem{identity}{Identity}
\newtheorem{define}{Definition}
\newtheorem{cor}{Corrolary}

\newcommand{\ket}[1]{\vert #1\rangle}
\newcommand{\e}{{\rm e}}

\theoremstyle{remark}
\newtheorem*{rmk}{Remark}
\newtheorem{prpty}{Property}
\usepackage{graphicx}

\begin{document}

\title{Quantum algorithm for time-dependent Hamiltonian simulation\\by permutation expansion}
\author{Yi-Hsiang Chen}
\affiliation{Information Sciences Institute, University of Southern California, Marina del Rey, CA 90292, USA}
\affiliation{Department of Physics and Astronomy, and Center for Quantum Information Science \& Technology,University of Southern California, Los Angeles, California 90089, USA}
\author{Amir Kalev}
%\email{amirk@isi.edu}
\affiliation{Information Sciences Institute, University of Southern California, Arlington, VA 22203, USA}
\author{Itay Hen}
%\email{itayhen@isi.edu}
\affiliation{Information Sciences Institute, University of Southern California, Marina del Rey, CA 90292, USA}
\affiliation{Department of Physics and Astronomy, and Center for Quantum Information Science \& Technology,University of Southern California, Los Angeles, California 90089, USA}

\begin{abstract}
\noindent We present a quantum algorithm for the dynamical simulation of time-dependent Hamiltonians. Our method involves expanding the interaction-picture Hamiltonian as a sum of generalized permutations, which leads to an integral-free Dyson series of the time-evolution operator. Under this representation, we perform a quantum simulation for the time-evolution operator  by means of the linear combination of unitaries technique. We optimize the time steps of the evolution based on the  Hamiltonian's dynamical characteristics, leading to a gate count that scales with an $L^1$-norm-like scaling with respect only to the norm of the interaction Hamiltonian, rather than that of the total Hamiltonian. 
We demonstrate that the cost of the algorithm is independent of the Hamiltonian's frequencies, implying its advantage for systems with highly oscillating components, and for time-decaying systems the cost does not scale with the total evolution time asymptotically. In addition, our algorithm retains the near optimal $\log(1/\epsilon)/\log\log(1/\epsilon)$ scaling with simulation error $\epsilon$.   
\end{abstract}

\maketitle

\section{Introduction}
The problem of simulating quantum systems, whether it is to study their dynamics, or to infer their salient equilibrium properties, was the original motivation for quantum
computers~\cite{Feynman:QC} and remains one of their major potential applications~\cite{Reiher7555,PhysRevX.8.011044}. Classical
algorithms for this problem are known to be grossly inefficient. Nonetheless, a significant fraction of the world's computing power
today is spent on solving instances of this problem --- a reflection on their importance~\cite{sc1,sc2,sc3}. 
%Important applications are numerous, in quantum chemistry, physics and material science, e.g.,  the discovery of new materials, the design of chemical processes and the simulation of quantum many-body condensed matter systems. 
%Unlike standard computers, which struggle with such simulations, quantum computers offer an efficient representation of many-body quantum systems, with a number of qubits that scales only logarithmically with the dimension of the system. With recent breakthroughs in hardware development leading to the fabrication of what seems to be scalable quantum technologies, the development of quantum algorithms for simulating quantum many-body systems has taken center stage and the resource efficiency with which quantum systems can be simulated via quantum circuits is now rightfully recognized as a task of critical significance to the success of practical quantum computing.

An important class of quantum simulations that is known to be particularly challenging, and is the focus of this work, is that of time-dependent quantum processes, which are at the heart of many important quantum phenomena. These include for example quantum control schemes~\cite{td1}, transition states of chemical reactions~\cite{timeDep2} analog quantum computers such as quantum annealers~\cite{farhi_quantum_2001} and the quantum approximate optimization algorithm~\cite{qaoa}. Devising state-of-the-art resource efficient quantum algorithms to simulate these types of processes on quantum circuits is therefore a very worthy cause: it will allow for the studying of said phenomena in a controllable and vastly more illuminating manner.

In the literature, a number of quantum algorithms designed to simulate the dynamics of time-dependent quantum many-body Hamiltonians already exist.  However, most of them are variants of algorithms that suit time-independent Hamiltonians but lack optimizations for dynamical ones. 
For example, Hamiltonians based on the Lie-Trotter-Suzuki decomposition were developed in Refs.~\cite{Wiebe_2011,PhysRevLett.106.170501}, where the complexity scales polynomially with error. More recent advances~\cite{Childs2014STOC,Berry2015} improve it to a logarithmic error scaling, which directly lead to applications in time-dependent Hamiltonian simulations~\cite{low2019hamiltonian,Berry2019}. A recent study by Berry {\it et al.}~\cite{Berry2020timedependent} improves the Hamiltonian scaling to $L^1-$norm, by considering the dynamical properties of the time-dependent Hamiltonian. However, these mostly comprise of slicing the dynamics into a sequence of `quasi-static' steps, each of which implementing a static quantum simulation module. In addition, all the above-mentioned algorithms assume a time-dependent oracle --- a straightforward but not necessarily practical assumption that can obscure the true complexity of the simulation when physical models are considered.
%"} More recent advances providing complexity logarithmic in the error~\cite{Childs2014STOC,Berry2015} mention that their techniques can be generalized to time-dependent scenarios, but do not analyze this case in detail. The most recent algorithms~\cite{QSP,low2019hamiltonian} are not directly applicable to the time-dependent case. In all of the above, switching to time-dependent simulations is not done optimally and (loosely speaking) comprises mainly of slicing the dynamics into a sequence of `quasi-static' steps, each of which implementing a static quantum simulation module~\cite{Berry2019}. 

The sub-optimality that characterizes existing quantum algorithms can be attributed mainly to the fact that the time-evolution operator for time-dependent Hamiltonians is a more intricate entity than its time-independent counterpart (this matter is discussed in more detail below): While in the time-independent case the Schr{\"o}dinger equation can be formally integrated, the time-evolution unitary operator for time-dependent systems is given in terms a Dyson series~\cite{Dyson49} --- a perturbative expansion, wherein each summand is a multi-dimensional integral over a time-ordered product of the (usually interaction-picture) Hamiltonian at different points in time. These time-ordered integrals pose multiple algorithmic and implementation challenges.  

In this paper, we provide a quantum algorithm for simulating a time-dependent Hamiltonian dynamics. This algorithm invokes a separation of the Hamiltonian $H(t)$ into a sum of a static diagonal part $H_0$ and a dynamical part $V(t)$, i.e., $H(t)=H_0+V(t)$, and switches to the interaction-picture with respect to $H_0$.  The target evolution operator becomes a product of an interaction-picture unitary $U_I(t)$ followed by a diagonal unitary $\e^{-iH_0t}$ that can be simulated efficiently. The interaction Hamiltonian $V(t)$ is expanded as a sum of generalized permutations, and the resulting Dyson series of the evolution operator $U_I(t)$ becomes an integral-free representation~\cite{kalev2020integralfree} with the notion of divided differences, which is a well-studied quantity~\cite{deBoor2005,davis1975interpolation,Mccurdy1980,GUPTA2020107385,Mccurdy1984,Zivcovich2019}.  The divided differences have an intuition of discretized derivatives and is closely related to polynomial interpolations~\cite{deBoor2005}. We refer the reader to Appendix \ref{DDapendix} for a short summary of the notion of the divided differences. Under this representation, we use the LCU method~\cite{Berry2015} to simulate $U_I(t)$ with a truncated Dyson series. We find a partitioning scheme that determines the duration of the time steps along the simulation. Following this procedure, in general, each time interval has a different duration which is determined form the Hamiltonian's dynamical characteristics and can lead to substantially fewer number of steps as compared to using identical-length simulation segments, typically used in quantum simulation algorithms. We analyze the implementation gate and qubit costs and discuss the circumstances under which our simulation algorithm provides improvements over the state-of-the-art. Specifically, our algorithm is independent of the oscillation frequencies of the Hamiltonian.  This is in stark contrast to existing algorithms which have dependence on $||d{H}(t)/dt||$, which grows with oscillation rates. Another class of Hamiltonians for which our algorithm is preferred over others is those with exponential decays. We show that for these systems, our algorithms requires asymptotically a finite number of steps which does not scale with the evolution time, leading in turn to an exponential saving comparing to the linear scaling in existing approaches. Moreover, the cost with Hamiltonian norm only mainly depends on the interaction Hamiltonian $V(t)$ and not the total Hamiltonian $H(t)$~\cite{Berry2020timedependent}. This also indicates an advantage of the algorithm when the time-dependent Hamiltonian is dominant by a static part.

The paper is organized as follows. In Sec.~\ref{ODExpansion}, we review the permutation expansion method that leads to an integral-free representation for the Dyson series, as introduced in Ref.~\cite{kalev2020integralfree}. In Sec.~\ref{simulationalgorithm}, we present in detail the simulation algorithm that combines the integral-free expression of the evolution operator with the LCU method, and analyze the circuit costs. We highlight the main advantages of our algorithm in Sec.~\ref{advantages}. In Sec.~\ref{generaltimedependence}, we address the cases when the exponential-sum expansion of the time-dependence is not exact and estimate the error that stems from a finite sum approximation. Finally, we give a brief summary for our methods and results in Sec.~\ref{sec:conc}.
 
\section{Permutation expansion method for time-dependent Hamiltonians}\label{ODExpansion}
In this section, we briefly describe the integral-free Dyson series expression of the evolution operator, derived from a permutation expansion of the time-dependent Hamiltonian~\cite{kalev2020integralfree}. 
Without loss of generality~\cite{pmr}, we expand a general time-dependent Hamiltonian in terms of products of time-dependent diagonal matrices, $D_i(t)$,  and permutation operators, $P_i$, i.e.,
\begin{equation}
H(t)=\displaystyle\sum_{i=0}^{M} D_i(t)P_i, \label{perm1}
\end{equation}
where  $P_0\equiv \mathds{1}$. This decomposition can be done efficiently as long as $M$ scales polynomially with $\log d$, where $d$ is the dimension  of the Hamiltonian. We decompose each diagonal matrix into a finite sum of exponential functions, i.e., 
\begin{equation}
D_i(t)=\sum_{k=1}^{K_i}\exp\left(\Lambda^{(k)}_i t\right) D_i^{(k)}, \label{perm2}
\end{equation}
where $\Lambda^{(k)}_i$ and $D_i^{(k)}$ are complex diagonal matrices with diagonal elements being
\begin{align}
&\lambda^{(k)}_{i,z}\equiv \langle z| \Lambda^{(k)}_i|z\rangle, \\
&d_{i,z}^{(k)}\equiv \langle z| D_i^{(k)} |z\rangle,
\end{align}
 in some basis $\{|z\rangle\}$ (the basis in which $D_0$ is diagonal) and $K_i$ indicates the number of terms in the exponential decomposition for $D_i(t)$. This can be done for many cases when the time dependencies are simple combinations of exponential terms. For simplicity we assume here that the $K_i$'s are finite, and address the most general time dependence in detail in Sec.~\ref{generaltimedependence} and refer to various algorithms~\cite{Beylkin2005,Beylkin2010,Braess2009,Wiscombe1977,Norvidas2010} for efficiently finding an exponential sum approximation of a function. 
 
For a lighter notation, we set $K_i=K$ for all $i$. %(for example, by filling up zero matrices up to $K$ terms). 
 We can evaluate the time-evolution operator $U(t)$ corresponding to $H(t)$ as
\begin{align}
U(t)&\equiv \mathcal{T}\text{exp}\left[-i\int_0^t H(t')dt' \right] \nonumber\\
&= \sum_{q=0}^{\infty}(-i)^q\int^t_0d\tau_q\cdots\int^{\tau_2}_0 d\tau_1 H(\tau_q)\cdots H(\tau_1) \\
&=\sum_{q=0}^{\infty}\sum_{\bold{i}_q}\sum_{\bold{k}_q}(-i)^q \int^t_0d\tau_q\cdots\int^{\tau_2}_0 d\tau_1 \exp\left(\Lambda^{(k_q)}_{i_q}\tau_q\right)D_{i_q}^{(k_q)}P_{i_q}\cdots \exp\left(\Lambda^{(k_1)}_{i_1}\tau_1\right)D_{i_1}^{(k_1)}P_{i_1}, \label{texpand}
\end{align}
where $\bold{i}_q=\{i_q,\cdots,i_1\}$ and $\bold{k}_q=\{k_q,\cdots,k_1\}$ are multi-indices. The action of $U(t)$ on a basis vector $|z\rangle$ is 
\begin{align}
U(t)|z\rangle&=\sum_{q=0}^{\infty}\sum_{\bold{i}_q}\sum_{\bold{k}_q}(-i)^q \int^t_0d\tau_q\cdots\int^{\tau_2}_0 d\tau_1 \exp\left(\Lambda^{(k_q)}_{i_q}\tau_q\right)D_{i_q}^{(k_q)}P_{i_q}\cdots \exp\left(\Lambda^{(k_1)}_{i_1}\tau_1\right)D_{i_1}^{(k_1)}P_{i_1}|z\rangle \label{LambdaD}\\\nonumber
&=\sum_{q=0}^{\infty}\sum_{\bold{i}_q}\sum_{\bold{k}_q}(-i)^q \int^t_0d\tau_q\cdots\int^{\tau_2}_0 d\tau_1 \exp\left(\lambda^{(k_q)}_{i_q,z_{\bold{i}_q}}\tau_q+\cdots+\lambda^{(k_1)}_{i_1,z_{\bold{i}_1}}\tau_1\right)d_{i_q,z_{\bold{i}_q}}^{(k_q)}\cdots d_{i_1,z_{\bold{i}_1}}^{(k_1)}P_{i_q}\cdots P_{i_1}|z\rangle \\\nonumber
&=\sum_{q=0}^{\infty}\sum_{\bold{i}_q}\sum_{\bold{k}_q}(-i)^q \int^t_0d\tau_q\cdots\int^{\tau_2}_0 d\tau_1 \exp\left(\lambda^{(k_q)}_{i_q,z_{\bold{i}_q}}\tau_q+\cdots+\lambda^{(k_1)}_{i_1,z_{\bold{i}_1}}\tau_1\right)d_{\bold{i}_q,z}^{(\bold{k}_q)}P_{\bold{i}_q}|z\rangle,
\end{align}
where $|z_{\bold{i}_j}\rangle\equiv P_{i_j}\cdots P_{i_1}|z\rangle$ with $j$ ranging from 1 to $q$, and $\lambda^{(k_j)}_{i_j,z_{\bold{i}_j}}\left(d_{i_j,z_{\bold{i}_j}}^{(k_j)}\right)$ is the $z_{\bold{i}_j}$th diagonal element of $\Lambda^{(k_j)}_{i_j}\left(D^{(k_j)}_{i_j}\right)$. $P_{\bold{i}_q}$ is a shorthand of $P_{i_q}\cdots P_{i_1}$, and similarly $d_{\bold{i}_q,z}^{(\bold{k}_q)}\equiv d_{i_q,z_{\bold{i}_q}}^{(k_q)}\cdots d_{i_1,z_{\bold{i}_1}}^{(k_1)}$. Figure \ref{perm_fig} illustrates the accumulative actions of $D^{(k)}_{i}P_{i}$ on a basis vector $|z\rangle$. %In Eq. (\ref{LambdaD}), $D^{(k)}_{i}$ is replaced by another diagonal matrix, $\exp(\Lambda^{(k)}_{i}\tau)D^{(k)}_{i}$. 

\begin{figure}
\centering
\includegraphics[width=\textwidth]{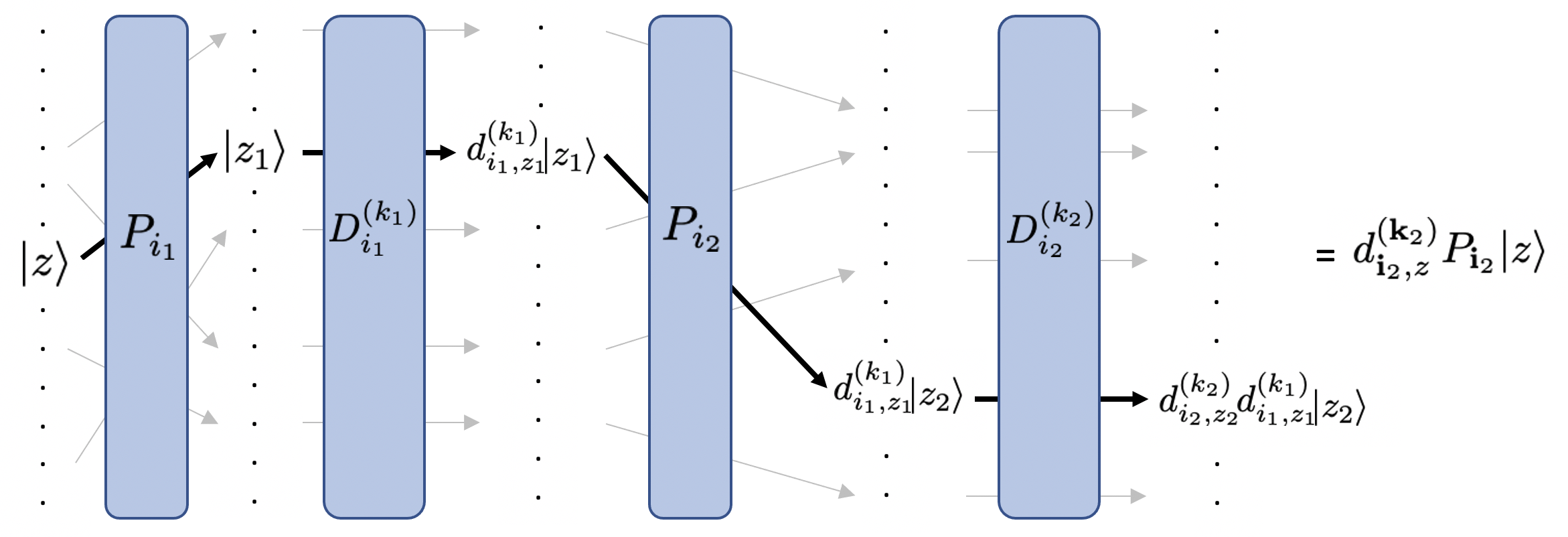}
\caption{The actions of a sequence of generalized permutations. This figure gives a pictorial illustration on how the elements of the diagonal matrices are picked up when interleaving with permutations. In this example, we have $q=2$.}
\label{perm_fig}
\end{figure}

To proceed, we use the following identity to simplify the expression in terms of divided differences. It is a variant of Hermite-Genocchi formula~\cite{deBoor2005} applying to the exponential function.
\begin{identity}\label{identity1}
For $\lambda_1,\cdots, \lambda_q\in \mathbb{C}$,
\begin{equation}
\int^1_0 ds_q\cdots \int^{s_2}_0ds_1 {\rm e}^{(\lambda_1 s_1+\cdots+\lambda_q s_q)}={\rm e}^{[x_1,\cdots,x_q,0]},
\end{equation}
where $x_j=\sum_{l=j}^q \lambda_l$ and ${\rm e}^{[x_1,\cdots,x_q,0]}$ is the divided difference of the exponential function with inputs $x_1,\cdots,x_q,0$. \emph{(The case with $q=1$ can be shown by explicit integration, and the identity follows by induction. For more details, see Ref.~\cite{kalev2020integralfree}.)}
\end{identity}

With this property, the multi-dimensional integration in the time-evolution operator can be simplified as
\begin{align}
U(t)|z\rangle&=\sum_{q=0}^{\infty}\sum_{\bold{i}_q}\sum_{\bold{k}_q} (-i)^q\int^t_0d\tau_q\cdots\int^{\tau_2}_0 d\tau_1 \exp\left(\lambda^{(k_q)}_{i_q,z_{\bold{i}_q}}\tau_q+\cdots+ \lambda^{(k_1)}_{i_1,z_{\bold{i}_1}}\tau_1\right) d^{(\bold{k}_q)}_{\bold{i}_q,z} P_{\bold{i}_q}|z\rangle \nonumber\\
&=\sum_{q=0}^{\infty}\sum_{\bold{i}_q}\sum_{\bold{k}_q} (-it)^q \int^1_0 ds_q\cdots\int^{s_2}_0 ds_1 \exp\left[t\left(\lambda^{(k_q)}_{i_q,z_{\bold{i}_q}}s_q+\cdots+ \lambda^{(k_1)}_{i_1,z_{\bold{i}_1}}s_1\right)\right]d^{(\bold{k}_q)}_{\bold{i}_q,z} P_{\bold{i}_q}  |z\rangle \nonumber\\
&=\sum_{q=0}^{\infty}\sum_{\bold{i}_q}\sum_{\bold{k}_q} (-i)^q \e^{t[x_1,x_2,\cdots, x_q,0]} d^{(\bold{k}_q)}_{\bold{i}_q,z}  P_{\bold{i}_q}|z\rangle, \label{form1}
\end{align}
where $x_j=\sum_{l=j}^q \lambda^{(k_l)}_{i_l,z_{\bold{i}_l}}$. The second equality uses the change of variable $d\tau=tds$, and the last equality follows from Identity \ref{identity1} and the identity of $t^q \e^{[t x_0,\cdots, t x_q]}=\e^{t[x_0,\cdots,x_q]}$. By completing the basis, we get
\begin{equation}
U(t)=\sum_z U(t)|z\rangle\langle z|=\sum_z\sum_{q=0}^{\infty}\sum_{\bold{i}_q}\sum_{\bold{k}_q} (-i)^q \e^{t[x_1,x_2,\cdots, x_q,0]} d^{(\bold{k}_q)}_{\bold{i}_q,z}  P_{\bold{i}_q}|z\rangle\langle z|. \label{expansion1} 
\end{equation}
This is an integral-free expression for the unitary time-evolution operator of the time-dependent Hamiltonian $H(t)$. We will later approximate the unitary by truncating the series at some order $q=Q$ that scales as $\mathcal{O}\left(\frac{\text{log}(1/\epsilon)}{\text{loglog}(1/\epsilon)}\right)$~\cite{Berry2015}, where $\epsilon$ is the required accuracy. 

\section{Time-dependent Hamiltonian simulation algorithm}\label{simulationalgorithm}
A time-dependent Hamiltonian $H(t)$ can be expressed as a sum of two Hamiltonians---a time-independent $H_0$ and a dynamical $V(t)$, i.e.,
\begin{equation}
H(t)=H_0+V(t).
\end{equation}
In many practical models, $H_0$ represents a static and simple Hamiltonian that is often diagonal in a known basis (which we will identify with the computational basis). Hence, hereafter, we assume that $H_0$ is a diagonal operator with real diagonal elements. The $V(t)$ component represents the nontrivial interactions between subsystems. Assume\footnote{For the most general cases, one can set $H_0=0$.} $H_0$ is diagonal in the computational basis $\{|z\rangle\}$. We switch to the interaction picture, i.e., 
\begin{equation}
\frac{d}{dt}|\psi(t)\rangle=-i H(t)|\psi(t)\rangle \to \frac{d}{dt}|\psi_I(t)\rangle=-i H_I(t)|\psi_I(t)\rangle,
\end{equation}
where 
\begin{equation}
|\psi_I(t)\rangle= \e^{iH_0t}|\psi(t)\rangle \ \ \text{and}\ \ H_I(t)=\e^{iH_0t}V(t)\e^{-iH_0t}.
\end{equation}
The Schr{\"o}dinger-picture unitary operator $U(t)$, satisfying $|\psi(t)\rangle=U(t)|\psi(0)\rangle$, is equivalent to a time-ordered matrix exponential followed by a diagonal unitary, i.e.,
\begin{equation}
U(t)=\e^{-iH_0t} \mathcal{T}\exp\left[-i\int_0^t H_I(t')dt' \right]=\e^{-iH_0t} \mathcal{T}\exp\left[-i\int_0^t \e^{iH_0t'}V(t') \e^{-iH_0t'}dt' \right]. \label{U_int}
\end{equation}
Hence, the simulation of $U(t)=\e^{-iH_0t}U_I(t)$ consists of two parts---a complicated $U_I(t)$ and a simple diagonal unitary $\e^{-iH_0t}.$ The simulation of $\e^{-iH_0t}$ can be achieved with a gate cost that scales only linearly with the locality of $H_0$. When we write $H_0=\sum^{L}_{\gamma=0}J_{\gamma}Z_{\gamma}$, where each $Z_{\gamma}$ is some tensor product of (single-qubit) Pauli-$Z$ operators acting on at most $d$ qubits, it can be shown that the gate cost scales as $\mathcal{O}(Ld)$~\cite{NielsenChuang,kalev2020simulating}.  Therefore, the main focus of our simulation is on $U_{I}$. 

We next provide an overview of the simulation algorithm in Sec.~\ref{overview}. In Sec.~\ref{algorithm}, we incorporate the LCU framework with the permutation expansion method. Sec.~\ref{stateprep} provides the state preparation operation and Sec.~\ref{circuitcost} evaluates the simulation cost for the whole procedure. 

\subsection{An overview of the algorithm}\label{overview}
Our proposed simulation algorithm consists of a permutation expansion procedure for $U_I$ and the LCU method for the quantum simulation. In Sec.~\ref{algorithm}, we explain in detail the essential ingredients for merging these two approaches. Before delving into technical details, we provide an overview of the algorithm in this section. 

Given a time-dependent Hamiltonian $H(t)$, we first decompose $H(t)$ into a sum of a static diagonal term $H_0$ (if exists) and a dynamical term $V(t)$. We switch to an interaction picture so that the target unitary evolution $U(T)$ over a period $T$ becomes
\begin{equation}
U(T)\equiv \mathcal{T}\exp\left[-i\int_0^T H(t)dt \right]=\e^{-iH_0T} \mathcal{T}\exp\left[-i\int_0^T \e^{iH_0t}V(t) \e^{-iH_0t}dt \right]\equiv \e^{-iH_0T}U_I(T).
\end{equation}
Therefore, the simulation of $U(T)$ is equivalent to applying $U_I(T)$ followed by $\e^{-iH_0T}$. Since the diagonal unitary $\e^{-iH_0T}$ can be efficiently simulated, we focus on $U_I(T)$ hereafter. 

Let us expand $V(t)$ as a sum of permutations as 
\begin{equation}
V(t)=\displaystyle\sum_{i=0}^{M} D_i(t)P_i, \label{Vexpand1}
\end{equation}
where $P_i$ are permutations ($P_0\equiv \mathds{1}$) and $D_i(t)$ are some diagonal matrices that are expressed as exponential sums, i.e.,
\begin{equation}
D_i(t)=\sum_{k=1}^{K}\exp\left(\Lambda^{(k)}_i t\right) D_i^{(k)}. \label{Vexpand2}
\end{equation}
$\Lambda^{(k)}_i$ and $D_i^{(k)}$ are some complex diagonal matrices. Partition $U_I(T)$ into $r$ segments $U_I(T,t_{r-1})\cdots U_I(t_1,0)$, whose respective durations $\Delta t_w$, $w=0,\ldots, r-1$ are determined by the partitioning scheme given in Sec.~\ref{algorithm} and the time markers $t_w$,  are defined as $t_w=\sum_{l=0}^{w-1} \Delta t_{l}$. The total number of steps is denoted as $r$. The evolution operator from $t_w$ to $t_w+\Delta t_w$ is expressed as
\begin{equation}
U_I(t_w+\Delta t_w, t_w)=\sum_{q=0}^{\infty}\sum_{\bold{i}_q}\sum_{\bold{k}_q}\sum_{x=\pm}(-i)^q \frac{\left(\frac{\e^{\Delta t_w \lambda}-1}{\lambda}\right)^q }{2q!}\Gamma^{(\bold{k}_q)}_{\bold{i}_q}(t_w) P_{\bold{i}_q} \Phi^{(\bold{k}_q,w)}_{\bold{i}_q,x},
\end{equation}
where we denote 
\begin{equation}
\Gamma^{(\bold{k}_q)}_{\bold{i}_q}(t_w)= \left|\left| D_{i_q}^{(k_q)}\right|\right|_{\max}\e^{t_w \lambda_{(i_q,k_q)}}\cdots\left|\left| D_{i_1}^{(k_1)}\right|\right|_{\max}\e^{t_w \lambda_{(i_1,k_1)}},
\end{equation}
where $||\cdot||_{\max}$ the max norm, $\lambda_{(i,k)}=\max_{z}\Re(\langle z|\Lambda^{(k)}_i|z\rangle)$ the maximum real part of $\Lambda^{(k)}_i$ and $\lambda=\max_{i,k}\{\lambda_{(i,k)}\}$. Here, $\Phi^{(\bold{k}_q,w)}_{\bold{i}_q,\pm}$ are some diagonal unitaries as derived later in Eq. (\ref{Uexpand}) and each $P_{\bold{i}_q}$ is a unique product of permutations. Note that the above evolution operators are given as a linear combination of unitaries (LCU). We provide a review for the LCU method in Appendix~\ref{LCUrecap}. We set the truncation order  $Q$ to be\footnote{An exact truncation order that guarantees the accuracy $\epsilon$ is $\left\lceil{\frac{\text{ln}(2r/\epsilon)}{W\left(\frac{\text{ln}(2r/\epsilon)}{e\text{ln}2}\right)}-1}\right\rceil$, where $W(\cdot)$ is the Lambert W-function.}
\begin{equation}
Q=\mathcal{O}\left(\frac{\log (r/\epsilon)}{\log\log (r/\epsilon)}\right),
\end{equation}
where $\epsilon$ is the overall simulation accuracy. 

To implement the LCU routine for each $U_I(t_w+\Delta t_w, t_w)$, we require preparing a state 
\begin{align}
|\psi_0\rangle=\frac{1}{\sqrt{s}}\sum_{q=0}^{Q}\sum_{\bold{i}_q}\sum_{\bold{k}_q}\sum_{x=0,1}\sqrt{\frac{\left(\frac{\e^{\Delta t_w \lambda}-1}{\lambda}\right)^q\Gamma^{(\bold{k}_q)}_{\bold{i}_q}(t_w)}{2q!}} |\bold{i}_q\rangle|\bold{k}_q\rangle|x\rangle,
\end{align}
where $|\bold{i}_q\rangle$ represents $Q$ quantum registers that each has dimension $M$ and $|\bold{k}_q\rangle$ represents $Q$ quantum registers that each has dimension $K$, and $s$ is the normalization factor. Following the same notation in Sec.~\ref{LCUrecap}, let us denote the state preparation unitary as $B$, i.e., $B|0\rangle^{\otimes 2Q+1}=|\psi_0\rangle$ (B is explicitly given in Sec.~\ref{stateprep}). Let us denote $V_c$ the control unitary such that
\begin{equation}
V_c|\bold{i}_q\rangle|\bold{k}_q\rangle|x\rangle |\psi\rangle=|\bold{i}_q\rangle|\bold{k}_q\rangle|x\rangle (-i)^qP_{\bold{i}_q} \Phi^{(\bold{k}_q,w)}_{\bold{i}_q,x}|\psi\rangle.
\end{equation}
The Oblivious Amplitude Amplification (OAA) involves interleaving the operator $W=(B^{\dagger}\otimes I)V_c(B\otimes I)$ as 
\begin{equation}
A=-WRW^{\dagger}RW,
\end{equation}
where $R\equiv I-2(|0\rangle\langle 0|\otimes I)$. For each piece of the unitary, we implement $A$ on the extended system $|0\rangle^{\otimes (2Q+1)} |\psi\rangle$. By construction, we have
\begin{equation}
\left|\left|A|0\rangle^{\otimes (2Q+1)} |\psi\rangle-|0\rangle^{\otimes (2Q+1)} U_I(t_w+\Delta t_w,t_w)|\psi\rangle \right|\right|=\mathcal{O}\left(\frac{\epsilon}{r}\right).
\end{equation}
This means that applying $A$ effectively performs the unitary $U_I(t_w+\Delta t_w,t_w)$ on the main system $|\psi\rangle$, with error $\mathcal{O}(\epsilon/r)$. Combining $r$ pieces of the procedure, it effectively simulates $U_I(T)$ with overall error $\mathcal{O}(\epsilon)$, i.e.,
\begin{equation}
\left|\left|A_{r-1} \cdots A_1A_0|0\rangle^{\otimes (2Q+1)} |\psi\rangle-|0\rangle^{\otimes (2Q+1)} U_I(T)|\psi\rangle \right|\right|=\mathcal{O}\left(\epsilon\right),
\end{equation}
where $A_w$ are the OAA operators for the corresponding piece of evolution. This implies that applying the sequence of $A$s followed by the circuit for $\e^{-iH_0T}$ can approach the action of $U(T)$ to an arbitrary accuracy.

\subsection{Permutation expansion for $U_I(t)$}\label{algorithm}
In this section, we give a thorough introduction of the permutation expansion in the Dyson series and the conditions arisen from implementing the LCU method. We focus on addressing the interaction-picture unitary $U_I(t)$, i.e., the time-ordered operator in Eq. (\ref{U_int}). Using the expansions introduced in Eqs.~(\ref{Vexpand1}) and (\ref{Vexpand2}), we get
\begin{align}
&U_I(t)\equiv \mathcal{T}\exp\left[-i\int_0^t \e^{iH_0t'}V(t') \e^{-iH_0t'}dt' \right] \nonumber \\
&=\sum_{q=0}^{\infty}(-i)^q\int^t_0d\tau_q\cdots\int^{\tau_2}_0 d\tau_1 \e^{iH_0\tau_q}V(\tau_q)\e^{-iH_0\tau_q}\cdots \e^{iH_0\tau_1}V(\tau_1)\e^{-iH_0\tau_1} \nonumber\\
&=\sum_{q=0}^{\infty}\sum_{\bold{i}_q}\sum_{\bold{k}_q}(-i)^q \int^t_0d\tau_q\cdots\int^{\tau_2}_0 d\tau_1  \e^{iH_0\tau_q}  \e^{\Lambda^{(k_q)}_{i_q} \tau_q} D_{i_q}^{(k_q)}P_{i_q}\e^{-iH_0\tau_q}\cdots \e^{iH_0\tau_1}  \e^{\Lambda^{(k_1)}_{i_1} \tau_1} D_{i_1}^{(k_1)}P_{i_1}\e^{-iH_0\tau_1},
\end{align}
We denote the basis in which $H_0$ is diagonal by $\{\ket{z}\}$ and its diagonal elements by $E_z=\langle z| H_0|z\rangle$. The action of $U_I(t)$ on a basis vector $|z\rangle$ becomes
\begin{align}
&U_I(t)|z\rangle= \sum_{q=0}^{\infty}\sum_{\bold{i}_q}\sum_{\bold{k}_q}(-i)^q \int^t_0d\tau_q\cdots\int^{\tau_2}_0 d\tau_1 \\\nonumber
&\times\exp\left[\left(iE_{z_{\bold{i}_q}}-iE_{z_{\bold{i}_{q-1}}}+\lambda^{(k_q)}_{i_q,z_{\bold{i}_q}}\right)\tau_q+\cdots+\left(iE_{z_{\bold{i}_1}}-iE_z+\lambda^{(k_1)}_{i_1,z_{\bold{i}_1}} \right)\tau_1\right] d^{(\bold{k}_q)}_{\bold{i}_q,z}P_{\bold{i}_q}|z\rangle, 
\end{align}
where $E_{z_{\bold{i}_j}}$ is the $z_{\bold{i}_j}$th diagonal element of $H_0$, i.e., $E_{z_{\bold{i}_j}}=\langle z_{\bold{i}_j}| H_0| z_{\bold{i}_j}\rangle$, and $\ket{z_{\bold{i}_j}}=P_{\bold{i}_j}|z\rangle$ with $P_{\bold{i}_j}=P_{i_j}\cdots P_{i_1}$. By Identity 1, this can be further simplified as
\begin{align}
U_I(t)|z\rangle=\sum_{q=0}^{\infty}\sum_{\bold{i}_q}\sum_{\bold{k}_q}(-i)^q \e^{t[x_1,\cdots,x_q,0]}d^{(\bold{k}_q)}_{\bold{i}_q,z}P_{\bold{i}_q}|z\rangle, \label{U_t}
\end{align}  
where 
\begin{equation}
x_j=i\left(E_{z_{\bold{i}_q}}-E_{z_{\bold{i}_{j-1}}}\right)+\sum_{l=j}^q \lambda^{(k_l)}_{i_l,z_{\bold{i}_l}}. \label{defofx}
\end{equation}

\subsection{The LCU routine}

To implement the LCU method for a quantum simulation of $U_I(T)$, we first decompose the overall simulation duration $T$ into $r$ pieces in sequence, i.e., 
\begin{equation}
U_I(T)=U_I(T, t_{r-1})U_I(t_{r-1},t_{r-2})\cdots U_I(t_1,0)=\prod^{r-1}_{w=0}U_I(t_w+\Delta t_w, t_w),
\end{equation}
where the operators in the product of the last equation are understood to be ordered, $t_{w+1}=t_w+\Delta t_w$ and $t_0\equiv 0$ and $t_r\equiv T$. The number of steps, $r$, and the step size, $\Delta t_w$, are to be determined. When acting on a computational basis state, each piece in the decomposition can be written as
\begin{align}
&U_I(t_w+\Delta t_w, t_w)|z\rangle=\mathcal{T}\text{exp}\left[-i\int_{ t_w}^{t_w+\Delta t_w} H_I(t')dt' \right]|z\rangle \nonumber\\
&=\sum_{q=0}^{\infty}\sum_{\bold{i}_q}\sum_{\bold{k}_q}(-i)^q \int^{t_w+\Delta t_w}_{t_w}d\tau_q\cdots\int^{\tau_2}_{t_w} d\tau_1\exp\Bigg[\sum_{l=1}^q\left(iE_{z_{\bold{i}_l}}-iE_{z_{\bold{i}_{l-1}}}+\lambda^{(k_l)}_{i_l,z_{\bold{i}_l}}\right)\tau_l\Bigg]d^{(\bold{k}_q)}_{\bold{i}_q,z}P_{\bold{i}_q}|z\rangle, \nonumber\\
&=\sum_{q=0}^{\infty}\sum_{\bold{i}_q}\sum_{\bold{k}_q}(-i)^q \exp \left[t_w\sum_{l=1}^q\left(iE_{z_{\bold{i}_l}}-iE_{z_{\bold{i}_{l-1}}}+\lambda^{(k_l)}_{i_l,z_{\bold{i}_l}} \right) \right] \nonumber\\
& \times \int^{\Delta t_w}_{0}d\tau'_q\cdots\int^{\tau'_2}_{0} d\tau'_1 \exp\Bigg[\sum_{l=1}^q\left(iE_{z_{\bold{i}_l}}-iE_{z_{\bold{i}_{l-1}}}+\lambda^{(k_l)}_{i_l,z_{\bold{i}_l}}\right)\tau'_l\Bigg]d^{(\bold{k}_q)}_{\bold{i}_q,z}P_{\bold{i}_q}|z\rangle \nonumber\\
&=\sum_{q=0}^{\infty}\sum_{\bold{i}_q}\sum_{\bold{k}_q}(-i)^q \exp \left[t_w\sum_{l=1}^q\left(iE_{z_{\bold{i}_l}}-iE_{z_{\bold{i}_{l-1}}}+\lambda^{(k_l)}_{i_l,z_{\bold{i}_l}} \right) \right]\e^{\Delta t_w[x_1,x_2,\cdots, x_q,0]} d^{(\bold{k}_q)}_{\bold{i}_q,z} P_{\bold{i}_q}|z\rangle \nonumber\\
& =\sum_{q=0}^{\infty}\sum_{\bold{i}_q}\sum_{\bold{k}_q}(-i)^q 
\e^{- i t_w (E_{z_{\bold{i}_{0}}}-E_{z_{\bold{i}_{q}}})}\e^{t_w\sum_{l=1}^q \lambda^{(k_l)}_{i_l,z_{\bold{i}_l}}} 
\e^{\Delta t_w[x_1,x_2,\cdots, x_q,0]} d^{(\bold{k}_q)}_{\bold{i}_q,z} P_{\bold{i}_q}|z\rangle, \label{U_t_dt}
\end{align}
which has the same form as Eq. (\ref{U_t}) except that the integration intervals are shifted (with $E_{z_{\bold{i}_{0}}}\equiv E_z$). We can denote 
\begin{equation}
d^{(\bold{k}_q)}_{\bold{i}_q,z}(t_w)=d^{(\bold{k}_q)}_{\bold{i}_q,z} \e^{t_w\sum_{l=1}^q \lambda^{(k_l)}_{i_l,z_{\bold{i}_l}}}, \nonumber
\end{equation}
which leads to
\begin{align}
&U_I(t_w+\Delta t_w, t_w)|z\rangle=
\sum_{q=0}^{\infty}\sum_{\bold{i}_q}\sum_{\bold{k}_q}(-i)^q 
\e^{- i t_w (E_{z_{\bold{i}_{0}}}-E_{z_{\bold{i}_{q}}})}
\e^{\Delta t_w[x_1,x_2,\cdots, x_q,0]} d^{(\bold{k}_q)}_{\bold{i}_q,z}(t_w) P_{\bold{i}_q}|z\rangle, \label{U_t_dt}
\end{align}
To formulate the above expression in terms of a linear combination of unitaries, we need to evaluate the norms of $\e^{\Delta t_w[x_1,x_2,\cdots, x_q,0]}$ and $d^{(\bold{k}_q)}_{\bold{i}_q,z}(t_w)$. The norm of $d^{(\bold{k}_q)}_{\bold{i}_q,z}(t_w)$ is bounded by
\begin{equation}\label{eq:Gamma-bound-d}
\left\vert d^{(\bold{k}_q)}_{\bold{i}_q,z}(t_w)\right\vert\leq ||D^{(k_q)}_{i_q}||_{\max}\e^{t_w \lambda_{(i_q,k_q)}}\cdots||D^{(k_1)}_{i_1}||_{\max} \, \e^{t_w\lambda_{(i_1,k_1)}} =\Gamma^{(\bold{k}_q)}_{\bold{i}_q}(t_w) \,.
\end{equation}
The norm of the $\e^{\Delta t_w[x_1,x_2,\cdots, x_q,0]}$ can be bounded by using the following identity. 
\begin{identity}\label{bound1}
For any $q+1$ complex values $x_0,\cdots,x_q\in \mathbb{C}$, 
\begin{equation}
\left| {\rm e}^{[x_0,\cdots,x_q]}\right|\leq {\rm e}^{[\Re(x_0),\cdots,\Re(x_q)]}=\frac{{\rm e}^{\xi}}{q!},
\end{equation}
where $\Re(\cdot)$ denotes the real part of an input and $\xi\in \left[\min\{\Re(x_0),\cdots,\Re(x_q)\},\max\{\Re(x_0),\cdots,\Re(x_q)\}\right]$.
\end{identity}
The proof can be found in Appendix \ref{DDapendix}. From Identity~\ref{bound1}, we show in Appendix \ref{forbounds} that
\begin{equation}
\left|\e^{\Delta t_w[x_1,x_2,\cdots, x_q,0]}\right|\leq \e^{\Delta t_w[q \lambda, (q-1)\lambda,\ldots,\lambda,0] }
=  \frac{1}{q!}\left(\frac{\e^{\Delta t_w \lambda }-1}{\lambda} \right)^q \equiv\frac{\widetilde{\Delta t}_w^q}{q!}
, \label{tbound2}
\end{equation}
where we denoted the quantity 
\begin{equation}
\widetilde{\Delta t}_w \equiv\frac{\e^{\Delta t_w \lambda }-1}{\lambda} . 
\end{equation}
With these bounds, the factors in the expansion form in Eq. (\ref{U_t_dt}) can be written as
\begin{align} \label{coefficientbounds}
 \e^{\Delta t_w[x_1,x_2,\cdots, x_q,0]} d^{(\bold{k}_q)}_{\bold{i}_q,z}(t_w) &=\frac{{\widetilde{\Delta t}_w}^q }{q!}\Gamma^{(\bold{k}_q)}_{\bold{i}_q}(t_w) \\\nonumber
 &\times \left(\frac{\e^{\Delta t_w[x_1,x_2,\cdots, x_q,0]}}{\widetilde{\Delta t}_w^q/q!}\frac{\e^{-i t_w  (E_{z_{\bold{i}_{0}}}-E_{z_{\bold{i}_{q}}})}d^{(k_q)}_{i_q,z}\cdots d^{(k_1)}_{i_1,z} \, \e^{t_w\sum_{l=1}^q \lambda^{(k_l)}_{i_l,z_{\bold{i}_l}}}}{\Gamma^{(\bold{k}_q)}_{\bold{i}_q}(t_w)}\right)  \\\nonumber
&= \frac{\widetilde{\Delta t}_w^q }{q!}\Gamma^{(\bold{k}_q)}_{\bold{i}_q}(t_w)\cos\left[\phi^{(\bold{k}_q)}_{\bold{i}_q,z}\right] \e^{i \theta^{(\bold{k}_q)}_{\bold{i}_q,z}} =\frac{\widetilde{\Delta t}_w^q }{2q!}\Gamma^{(\bold{k}_q)}_{\bold{i}_q}(t_w) \left( \e^{i \phi^{(\bold{k}_q)}_{\bold{i}_q,z}+i \theta^{(\bold{k}_q)}_{\bold{i}_q,z}}+\e^{-i \phi^{(\bold{k}_q)}_{\bold{i}_q,z}+i \theta^{(\bold{k}_q)}_{\bold{i}_q,z}}\right),
\end{align}
where
$$
\phi^{(\bold{k}_q)}_{\bold{i}_q,z}=\cos^{-1}\left[\left|\frac{\e^{\Delta t_w[x_1,x_2,\cdots, x_q,0]} d^{(\bold{k}_q)}_{\bold{i}_q,z}(t_w)}{\frac{\widetilde{\Delta t}_w^q }{q!}\Gamma^{(\bold{k}_q)}_{\bold{i}_q}(t_w)}\right|\right]
$$
and
$$
\theta^{(\bold{k}_q)}_{\bold{i}_q,z}=\arg{\left[\frac{\e^{\Delta t_w[x_1,x_2,\cdots, x_q,0]} d^{(\bold{k}_q)}_{\bold{i}_q,z}(t_w)}{\frac{\widetilde{\Delta t}_w^q }{q!}\Gamma^{(\bold{k}_q)}_{\bold{i}_q}(t_w)}\right]}.
$$
The evolution operator from $t_w$ to $t_w+\Delta t_w$ becomes
\begin{align}
U_I(t_w+\Delta t_w, t_w)&=\sum_z U_I(t_w+\Delta t_w, t_w)|z\rangle\langle z|  \nonumber\\
&=\sum_z\sum_{q=0}^{\infty}\sum_{\bold{i}_q}\sum_{\bold{k}_q}(-i)^q \frac{\widetilde{\Delta t}_w^q }{2q!}\Gamma^{(\bold{k}_q)}_{\bold{i}_q}(t_w) \left( \e^{i \phi^{(\bold{k}_q)}_{\bold{i}_q,z}+i \theta^{(\bold{k}_q)}_{\bold{i}_q,z}}+\e^{-i \phi^{(\bold{k}_q)}_{\bold{i}_q,z}+i \theta^{(\bold{k}_q)}_{\bold{i}_q,z}}\right)P_{\bold{i}_q} |z\rangle\langle z| \nonumber\\
&=\sum_{q=0}^{\infty}\sum_{\bold{i}_q}\sum_{\bold{k}_q}\sum_{x=\pm}(-i)^q \frac{\widetilde{\Delta t}_w^q }{2q!}\Gamma^{(\bold{k}_q)}_{\bold{i}_q}(t_w)P_{\bold{i}_q} \Phi^{(\bold{k}_q,w)}_{\bold{i}_q,x}, \label{Uexpand}
\end{align}
where $\Phi^{(\bold{k}_q,w)}_{\bold{i}_q,\pm}$ are diagonal unitaries with diagonal elements being $\e^{i\left(\pm  \phi^{(\bold{k}_q)}_{\bold{i}_q,z}+ \theta^{(\bold{k}_q)}_{\bold{i}_q,z}\right)}$.

To implement the LCU method for simulating $U_I(t_w +\Delta t_w, t_w)$, we require a preparation of the state 
\begin{align}
|\psi_0\rangle&=\frac{1}{\sqrt{s}}\sum_{q=0}^{Q}\sum_{\bold{i}_q}\sum_{\bold{k}_q}\sum_{x=0,1} \sqrt{\frac{\widetilde{\Delta t}_w^q }{2q!}\Gamma^{(\bold{k}_q)}_{\bold{i}_q}(t_w)}|i_1\rangle\cdots|i_q\rangle\otimes|0\rangle^{\otimes(Q-q)}|k_1\rangle\cdots|k_q\rangle\otimes|0\rangle^{\otimes(Q-q)}|x\rangle \nonumber\\
&\equiv\frac{1}{\sqrt{s}}\sum_{q=0}^{Q}\sum_{\bold{i}_q}\sum_{\bold{k}_q}\sum_{x=0,1}\sqrt{\frac{\widetilde{\Delta t}_w^q }{2q!}\Gamma^{(\bold{k}_q)}_{\bold{i}_q}(t_w)} |\bold{i}_q\rangle|\bold{k}_q\rangle|x\rangle, \label{psi0}
\end{align}
where $|\bold{i}_q\rangle$ represents $Q$ quantum registers that each has dimension $M$ and $|\bold{k}_q\rangle$ represents $Q$ quantum registers that each has dimension $K$. The normalization constant is
\begin{equation}
s=\sum_{q=0}^{Q}\sum_{\bold{i}_q}\sum_{\bold{k}_q}\sum_{x=0,1}\frac{\widetilde{\Delta t}_w^q }{2q!}\Gamma^{(\bold{k}_q)}_{\bold{i}_q}(t_w)\equiv\sum_{q=0}^{Q}\frac{(\Gamma(t_w) \widetilde{\Delta t}_w)^q }{q!}, \label{sexpand}
\end{equation} 
where we define the $\Gamma(t_w)$ as 
\begin{equation}
\Gamma(t_w) \equiv \sum_{i=0}^{M}\sum_{k=1}^K \left|\left| D^{(k)}_i \right|\right|_{\max} \e^{t_w \lambda_{(i,k)}} \,,
\end{equation}
and we note that $\Gamma(t_w)$ is an upper bound on the max-norm of the interaction Hamiltonian at time $t_w$, $\Gamma(t_w)\geq\Vert V(t_w)\Vert_{\max}$. The quantity $\Gamma(t_w) $ is related to the energy strength in a typical LCU setup \cite{Berry2015}. In Appendix \ref{AltLCU}, we provide an alternative way that uses a larger bound $\Gamma=MK \max_{\forall k,i}||D^{(k)}_i||_{\max}$, which leads to an exponential saving for the state preparation. We proceed with $\Gamma(t_w)$ hereafter.

The OAA step in the LCU method requires $s\approx 2$. This leads to 
\begin{equation}
\Gamma(t_w) \widetilde{\Delta t}_w =\Gamma(t_w) \frac{\e^{\Delta t_w \lambda}-1}{\lambda} = \text{ln} 2, \label{stepcondition}
\end{equation}
and Eq. (\ref{sexpand}) becomes a truncated Taylor expansion of 2 up to order $Q$, i.e., $2\approx\sum_{q=0}^Q\frac{(\text{ln}2)^q}{q!}$. If we require $|s-2|\leq \epsilon/r$, where $r$ is the total number of steps and $\epsilon$ is some positive number, then the simulation error for each $U_I(t_w+\Delta t_w,t_w)$ is also within $\epsilon/r$. The required truncation order with this accuracy scales as  
\begin{equation}
Q=\mathcal{O}\left(\frac{\log (r/\epsilon)}{\log\log (r/\epsilon)}\right).
\end{equation}

\subsubsection{Time partitioning and number of time steps}\label{timepartition}

The condition in Eq. (\ref{stepcondition}) imposes a constraint on the next step size $\Delta t_w$ given the current time $t_w$,
\begin{equation}
\Delta t_w  = \frac1{\lambda}\ln\left( 1+\frac{\lambda}{\Gamma(t_w)}\ln 2\right) \,. \label{stepcondition1}
\end{equation}
Remembering that $\Gamma(t_w)$ is a function of $t_w=\sum_{l=0}^{w-1} \Delta t_{l}$, this condition determines the schedule, as every $\Delta t_w$ is determined by the preceding time steps. 

Special care should be given when setting the last time step, as $\Delta t_w$ can become too large that exceeds the total desired evolution time $T$. Whenever $t_{w+1}$ is found to be greater than $T$ (or if the argument inside the $\ln(\cdot)$ is found to be negative), one should replace the bound $\Gamma(t_w)$ with a larger bound \hbox{$\tilde{\Gamma}(t_w) =\lambda \ln 2/ (\e^{\lambda \Delta t_w}-1)$} and set the final step $\Delta t_w=T-t_w$.

Let us now examine the dependence of $\Delta t_w$ on   $\Gamma(t_w)$ in order to determine a bound on the number of time steps (equivalently, number of repetitions) $r$ required for the execution of the entire time evolution. We distinguish between three cases. (i) When $\lambda=0$, we have $\Delta t_w\Gamma(t_w) = \ln 2$, similar to the time-independent case though we note that a vanishing maximal $\lambda$ could imply time-dependent oscillations as well. This can be seen by taking the $\lambda \to 0$ limit of Eq.~\eqref{stepcondition1}. (ii) 
In the case where $\lambda<0$, i.e., a system with a decaying $\Gamma(t_w)$, we have $\Delta t_w\Gamma(t_w) \geq \ln 2$, i.e., the time steps are longer than $\ln 2/\Gamma(t_w)$. Furthermore, the total number of steps $r$ is finite even for an arbitrarily large evolution time $T$. Note that since $\Gamma(t_w)$ approaches zero asymptotically, for a large enough time $t_{w*}$, we have $\Gamma(t_{w*})<|\lambda|\ln 2$, i.e., the argument inside the logarithm above becomes negative. This indicates it reaches the final step, i.e., the bound should be modified as \hbox{$\tilde{\Gamma}(t_{w*}) =\lambda \ln 2/ (\e^{\lambda \Delta t_{w*}}-1)$} and $\Delta t_{w*}=T-t_{w*}$ becomes the final step. 
(iii) In the case where $\lambda>0$ (an amplified $\Gamma(t_w)$), we have $\Gamma(t_w)\gg \lambda$  at large simulation times, $t_w$. From Eq.~(\ref{stepcondition1}), we have $\Delta t_w\to \ln2/\Gamma(t_w)$ in this limit.

We see that (for large enough simulation times) the time step $\Delta t_w$ is inversely proportional to $\Gamma(t_w)$ which upper-bounds the max-norm of the interaction Hamiltonian at time $t_w$. Therefore, we have $\sum^{r-1}_{w=0} \Gamma(t_w)\Delta t_w\gtrsim r \ln2$, which implies $r\lesssim \sum^{r-1}_{w=0} \Gamma(t_w)\Delta t_w/\ln2$.
 
It would be instructive to compare the above scaling with that of Ref.~\cite{Berry2020timedependent} in which the simulation algorithm is said to have an $L^1$-norm scaling, i.e., an algorithm cost scaling linearly with $\int_{o}^t  d\tau H_{\max}(\tau) $ up to logarithmic factors. Under a similar intuition, our algorithm has a {\it discretized} $L^1$-norm-like scaling with $\sum^{r-1}_{w=0} \Gamma(t_w)\Delta t_w$. However in our case, $\Gamma(t_w)$ is related to the norm of the interaction Hamiltonian. 
 
 \subsubsection{State preparation}\label{stateprep}
 
 In this subsection, we provide a procedure to prepare the state $|\psi_0\rangle$ given in Eq. (\ref{psi0}). First, we initialize a state $|0\rangle^{\otimes Q}|0\rangle^{\otimes Q}|0\rangle$, where each of the first $Q$ registers has dimension $M$ (responsible for $|\bold{i}_q\rangle$ part), each of the later $Q$ registers has dimension $K$ (responsible for $|\bold{k}_q\rangle$ part), and the last register is a qubit (for the cosine decomposition). For simplicity, we can perform a Hadamard gate on the last qubit and then omit its dependence for the following discussion. The next step is to create a state in following the form,  
\begin{equation}\label{midstate}
\frac{1}{\sqrt{s} } \displaystyle\sum_{q=0}^{Q}\sqrt{s_q} |1\rangle^{\otimes q}|0\rangle^{\otimes (Q-q)}|1\rangle^{\otimes q}|0\rangle^{\otimes (Q-q)},
\end{equation}
where $s_q\equiv \left(\Gamma(t_w)\widetilde{\Delta t}_w  \right)^q/q!$. For each $|1\rangle$ from the first $Q$ registers (the $|\bold{i}_q\rangle$ part) and the corresponding $|1\rangle$ in the later $Q$ registers (the $|\bold{k}_q\rangle$ part), we make 
\begin{equation}
|1\rangle|1\rangle\to \sum_{i=0}^{M}\sum_{k=1}^K\sqrt{\frac{||D^{(k)}_i ||_{\max}\e^{t_w \lambda_{(i,k)}} }{\Gamma(t_w)}}| i \rangle | k \rangle. 
\end{equation}
Then Eq.~(\ref{midstate}) becomes
\begin{equation}
\frac{1}{\sqrt{s}} \displaystyle\sum_{q=0}^{Q}\sqrt{s_q}\sum_{\bold{i}_q}\sum_{\bold{k}_q}\sqrt{\frac{\Gamma^{(\bold{k}_q)}_{\bold{i}_q}(t_w)}{(\Gamma(t_w))^q}}|\bold{i}_q\rangle|\bold{k}_q\rangle=\frac{1}{\sqrt{s} } \displaystyle\sum_{q=0}^{Q}\sum_{\bold{i}_q}\sum_{\bold{k}_q}\sqrt{\frac{(\widetilde{\Delta t}_w)^q}{q!}\Gamma^{(\bold{k}_q)}_{\bold{i}_q}(t_w)}|\bold{i}_q\rangle|\bold{k}_q\rangle, 
\end{equation}
which is the required  $|\psi_0\rangle$ in Eq.~(\ref{psi0}), when combined with $|x\rangle$.

Next, we provide a process that produces the state in Eq.~(\ref{midstate}). First, we perform a rotation that takes the first register in the $|\bold{i}_q\rangle$ part to
\begin{equation}
|0\rangle \to \frac{1}{\sqrt{s} }\left(|0\rangle + \sqrt{\displaystyle\sum_{q=1}^{Q}s_q}|1\rangle \right),
\end{equation}
and perform a control gate from the first register to the second (both in the $|\bold{i}_q\rangle$ part) such that
\begin{align}
&\frac{1}{\sqrt{s} }\left(|0\rangle + \sqrt{\displaystyle\sum_{q=1}^{Q}s_q}|1\rangle \right)|0\rangle  \nonumber\\
&\to \frac{1}{\sqrt{s} }\Bigg[|00\rangle + \sqrt{\sum_{q=1}^{Q}s_q}|1\rangle \frac{1}{\sqrt{\displaystyle\sum_{q=1}^Qs_q}} \Bigg(\sqrt{s_1} |0\rangle+\sqrt{\displaystyle\sum_{q=2}^Qs_q} |1\rangle\Bigg)  \Bigg] \nonumber\\
&= \frac{1}{\sqrt{s} }\left(|00\rangle + \sqrt{s_1}|10\rangle + \sqrt{\displaystyle\sum_{q=2}^Qs_q} |11\rangle \right).
\end{align}
Continuing this procedure for the rest of the registers in the $|\bold{i}_q\rangle$ part, the state becomes 
\begin{align}
|0\rangle^{\otimes Q} \to \frac{1}{\sqrt{s}} \displaystyle\sum_{q=0}^{Q}\sqrt{s_q} |1\rangle^{\otimes q}|0\rangle^{\otimes (Q-q)}.
\end{align}
At this step, we perform CNOT operations\footnote{Strictly speaking, they are not standard CNOTs but higher-dimensional operations that act like a CNOT on the first two levels.} from the first $Q$ registers ($|\bold{i}_q\rangle$ part) to the last $Q$ registers ($|\bold{k}_q\rangle$ part) correspondingly, e.g., perform a CNOT from the first register in the $|\bold{i}_q\rangle$ part to the first register in the $|\bold{k}_q\rangle$ part, and so on and so forth. Finally, we have
\begin{align}
\frac{1}{\sqrt{s} } \displaystyle\sum_{q=0}^{Q}\sqrt{s_q} |1\rangle^{\otimes q}|0\rangle^{\otimes (Q-q)} |0\rangle^{\otimes Q} \to \frac{1}{\sqrt{s} } \displaystyle\sum_{q=0}^{Q}\sqrt{s_q} |1\rangle^{\otimes q}|0\rangle^{\otimes (Q-q)}|1\rangle^{\otimes q}|0\rangle^{\otimes (Q-q)},
\end{align}
which gives Eq.~(\ref{midstate}) as required. The estimated gate cost for the preparation of $|\psi_0\rangle$ is $\mathcal{O}(QMK)$. More detail regarding the cost is provided in Sec. \ref{B_Vc}.

 \subsubsection{Implementation of the controlled unitaries}\label{cu}

The second ingredient of the LCU routine is the construction of the controlled operation
%Let us denote $V_c$ the control unitary such that
\begin{equation}\label{eq:CU}
V_c|\bold{i}_q\rangle|\bold{k}_q\rangle|x\rangle |\psi\rangle=|\bold{i}_q\rangle|\bold{k}_q\rangle|x\rangle (-i)^qP_{\bold{i}_q} \Phi^{(\bold{k}_q,w)}_{\bold{i}_q,x}|\psi\rangle.
\end{equation}
Taking an approach similar to that taken in Ref.~\cite{kalev2020simulating}, we first note that Eq.~\eqref{eq:CU} indicates that $V_c$ can be carried out in two steps: a controlled-phase operation ($V_{c\Phi}$) followed by a controlled-permutation operation ($V_{cP}$). 

The controlled-phase operation $V_{c \Phi}$ requires a somewhat intricate calculation of non-trivial phases. We therefore carry out the required algebra with the help of additional ancillary registers and then `push' the results into phases. The latter step is done by employing the unitary 
\begin{align}
U_{\text{ph}}|\varphi\rangle=\e^{-i \varphi}|\varphi\rangle \,,
\end{align}
whose implementation cost depends only on the precision with which we specify $\varphi$  and is independent of Hamiltonian parameters~\cite{NielsenChuang} (see Ref.~\cite{kalev2020simulating} for a complete derivation).
With the help of the (controlled) unitary transformation
\begin{equation}
\label{eq:F}
V_{\chi\phi}|\bold{i}_q\rangle|\bold{k}_q\rangle|x\rangle|z\rangle|0\rangle= |\bold{i}_q\rangle|\bold{k}_q\rangle|x\rangle|z\rangle|\chi_{{\bf i}_q}^{(z)} +(-1)^k\phi_{{\bf i}_q}^{(z)}\rangle \,,
\end{equation}
we can write
$V_{c{\Phi}}=V_{\chi\phi}^\dagger (\mathds{1} \otimes U_{\text{ph}}) V_{\chi\phi}$,
so that 
\begin{equation}\label{eq:ucphi}
V_{c{\Phi}}|\bold{i}_q\rangle|\bold{k}_q\rangle|x\rangle|z\rangle=|\bold{i}_q\rangle|\bold{k}_q\rangle|x\rangle\Phi_{{\bf i}_q}^{(k)}|z\rangle \,.
\end{equation}
Note that $V_{\chi\phi}$ sends computational basis states to computational basis states. We provide an explicit construction of  $V_{\chi\phi}$ in Ref.~\cite{kalev2020simulating}. 
We find that its gate cost is $ \mathcal{O}(QM(k_{od}+\log M)+QMK(C_D+C_{\Delta H_0}+C_{\Lambda}))$ and qubit cost is $\mathcal{O}(Q\log (MK)).$ Addiitonal details are provided in Sec. \ref{B_Vc}.
%respectively, where $k_{\rm od}$ is an upper bound on the `off-diagonal locality', i.e., the locality of the $P_i$'s~\cite{PhysRevA.52.3457,Berry1} and $C_{\Delta D_0}$ is the cost of calculating the change in diagonal energy due to the action of a permutation operator. %.$C_{D}$ is the cost of calculating a single $D_j$ matrix element

The construction of $V_{c P}$ is carried out by a repeated execution of the simpler unitary transformation \hbox{$U_p|i\rangle|z\rangle = |i\rangle P_i|z\rangle$}.  Recall that $P_i$ are the off-diagonal permutation operators that appear in the Hamiltonian. The gate cost of $U_p$ is therefore ${\cal O}(M (k_{\rm od} +\log M))$. Additional details may be found in  Ref.~\cite{kalev2020simulating}. 

\subsection{Algorithm cost}\label{circuitcost}
We next analyze the circuit costs for the permutation expansion algorithm. Recall that the simulation of $U(T)$ consists of two operations---$\e^{-iH_0T}$ and $U_I(T)$. The diagonal unitary $\e^{-iH_0T}$ can be implemented efficiently with a gate cost that scales linearly with the system size. To observe this, note that $H_0$ is a diagonal matrix with real diagonal elements and can be written as $H_0=\sum_{\gamma=0}^L J_{\gamma}Z_{\gamma}$, where each $Z_{\gamma}$ is a tensor product of Pauli-$Z$'s ($Z\otimes\cdots \otimes Z$) acting on at most $d$ qubits (weight-$d$ operators). Hence, we can write $\e^{-iH_0T}=\prod_{\gamma=0}^L \e^{-iJ_{\gamma}Z_{\gamma}T}$. Each $\e^{-iJ_{\gamma}Z_{\gamma}T}$ can be simulated using at most $2d$ CNOT gates with a single ancillary qubit. For example, let $Z_{\gamma}$ be a weight-$m$ ($m\leq d$) operator, then $\e^{-iJ_{\gamma}Z_{\gamma}T}$ can be implemented as 
\begin{figure}[h]
\centering
\includegraphics[width=6cm,keepaspectratio]{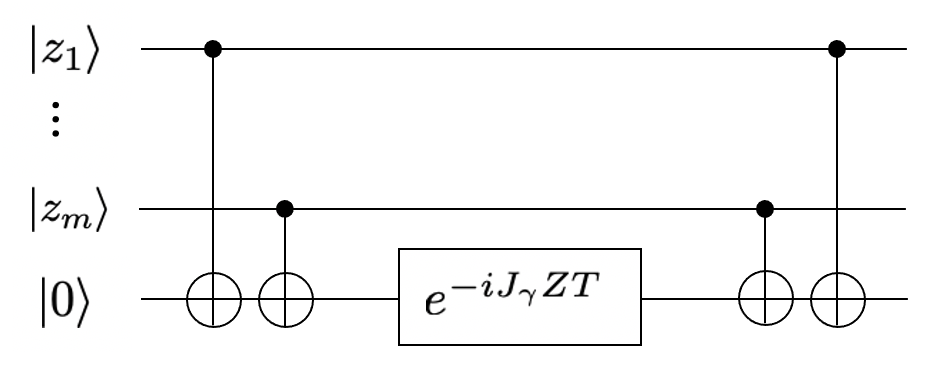}
\label{Zgate}
\end{figure}
\\
where  $|z_1\rangle \cdots |z_m\rangle$ are the qubits $Z_{\gamma}$ acts on and $|0\rangle$ is an ancillary qubit for extracting the phase. There are total $L$ such implementations for $\e^{-iH_0T}$. Therefore, the total gate cost is $\mathcal{O}(Ld)$ and the qubit cost is $\mathcal{O}(1)$. Since $L$ usually grows linearly with the system size, the gate cost also scales linearly.

\subsubsection{The cost for the state preparation and the controlled unitaries}\label{B_Vc}
The cost of implementing $U_I(T)$ resembles those in Ref.~\cite{kalev2020simulating}. The first ingredient is the preparation of state $|\psi_0\rangle$. Recall from Sec.~\ref{stateprep}, the operation that takes $|0\rangle^{\otimes Q}|0\rangle^{\otimes Q}|0\rangle$ to 
$\frac{1}{\sqrt{s} } \sum_{q=0}^{Q}\sqrt{s_q} |1\rangle^{\otimes q}|0\rangle^{\otimes (Q-q)}|1\rangle^{\otimes q}|0\rangle^{\otimes (Q-q)}$ has gate cost $\mathcal{O}(Q)$. The operation for $|1\rangle|1\rangle\to  \sum_{i=0}^{M}\sum_{k=1}^K\sqrt{\frac{||D^{(k)}_i ||_{\max}\e^{t_w \lambda_{(i,k)}} }{\Gamma(t_w)}}| i \rangle | k \rangle$  costs $\mathcal{O}(MK)$ \cite{1629135}. The total gate cost for the preparation of $|\psi_0\rangle$ (i.e., $B$) is $\mathcal{O}(QMK)$ (Lemma 8 in \cite{doi:10.1137/16M1087072}). In Appendix \ref{AltLCU}, we provide an alternative procedure that leads to a $\mathcal{O}(Q\log (MK))$ scaling for implementing $B$. The qubit cost in the state preparation is $\mathcal{O}(Q\log (MK)).$ 

The next component is the implementation of the control unitary $V_c$. As shown in~\cite{kalev2020simulating}, the gate cost of performing the control permutation $P_{\bold{i}_q}$ is $\mathcal{O}(QM(k_{od}+\log M))$, where $k_{od}$ is the ``locality,'' i.e., each permutation $P_i$ is a tensor product of at most $k_{od}$ Pauli-$X$ operators. The implementation of the control phase $\Phi^{(\bold{k}_q,w)}_{\bold{i}_q,x}$ involves the calculation of $d^{(\bold{k}_q)}_{\bold{i}_q,z}$ (the product of diagonal elements in the permutation expansion) and the divided differences (with $x_j$'s being the inputs). The cost of the former is $\mathcal{O}(QMKC_D)$, where $C_D$ is the cost of obtaining an element of $D^{(k)}_i$. The cost of later is $\mathcal{O}(QM(k_{od}+\log M)+QMK(C_{\Delta H_0}+C_{\Lambda}))$, where $C_{\Delta H_0}$ $(C_{\Lambda})$ is the cost of obtaining energy differences of $H_0$ (elements of $\Lambda^{(k)}_i$) [therefore, $C_{\Delta H_0}+C_{\Lambda}$ is the cost for obtaining the inputs $x_j$'s as defined in Eq. (\ref{defofx})]. The additional cost for the reversibility of the process scales as $\mathcal{O}(Q^2)$. A detailed discussion of the costs of $C_{ \Delta H_0}$ and $C_{\Lambda}$ may be found in Ref.~\cite{kalev2020simulating}. Combining these, we estimate the total cost for $V_c$ is 
 \begin{equation}
 \mathcal{O}(Q^2+QM(k_{od}+\log M)+QMK(C_D+C_{\Delta H_0}+C_{\Lambda})). \label{scale2}
 \end{equation} 

\subsubsection{	Overall cost of the algorithm}\label{r_cost}
The full simulation for $U_I(T)$ is a product of segments $U_I(t_w+\Delta t_w,t_w)$, where each segment is simulated by interleaving $B$ and $V_c$. The total number of segments, $r$, is determined by $T=\sum_{w=0}^{r-1}\Delta t_w$, where each $\Delta t_w$ is determined by partitioning scheme described in Sec.~\ref{timepartition}.  

As discussed above, the number of LCU applications $r$ can be upper-bounded by $r\lesssim \sum^{r-1}_{w=0} \Gamma(t_w)\Delta t_w/\ln2$ (in the long simulation time limit), which can be viewed as a discretized $L^1$-norm-like scaling with the norm of the non-static component of the Hamiltonian $V(t)$. 

Combining with the cost for simulating $\e^{-iH_0 T}$ and the cost for each step (\ref{scale2}), we conclude that at worst, the total gate cost scales as 
\begin{equation}
\mathcal{O}\left(r\left(Q^2+QM(k_{od}+\log M)+QMK(C_D+C_{\Delta H_0}+C_{\Lambda})\right)+Ld\right), \label{totalcost}
\end{equation}
and the qubit cost scales as 
\begin{equation}
\mathcal{O}(Q\log (MK)), \label{totalqubitcost}
\end{equation}
where $Q$ scales as $\mathcal{O}(\log (r/\epsilon)/\log\log(r/\epsilon))$.
For convenience, we provide a glossary of symbols in Table~\ref{tbl:glossary}. A summary of the gate and qubit costs of the simulation circuit and the various sub-routines used to construct it is given in Table~\ref{tbl:resource}. \vspace{0.3cm}
\begin{table*}[t!]
\begin{center}
\begin{tabular}{ |c|l|}
  \hline
  Symbol&Meaning\\
  \hline
  $M$ & the number of permutation expansion terms of the non-static Hamitonain, c.f., Eq.~\eqref{Vexpand1} \\
  $K$ & the length of exponential sum expansion, c.f., Eq.~\eqref{Vexpand2} \\
  $r$ & the number partitions, c.f., Sec.~\ref{timepartition}\\
  $Q$ & the series expansion truncation order, $Q={\cal O}\Bigl(\frac{\log (r/\epsilon)}{\log \log (r/\epsilon)}\Bigr)$\\
$k_{\rm od}$ & the upper bound on the locality %(\# of Pauli-Xs) 
of $P_i$  \\
$C_{D}$ & the cost of obtaining an element of $D^{(k)}_i$  \\
$C_{\Delta H_0}$ & the cost of obtaining energy differences of $H_0$    \\
$C_{\Lambda}$ & the cost of obtaining an element of $\Lambda^{(k)}_i$   \\
$L$ & the number of terms in the static Hamiltonian, i.e., $H_0=\sum_{\gamma=0}^L J_{\gamma}Z_{\gamma}$ \\
$d$ & the locality %(\# of Pauli-Zs) 
of $Z_{\gamma}$   \\
    \hline
 \end{tabular}
\end{center}
\caption{\label{tbl:glossary}{\bf Glossary of symbols.} }
\end{table*}

\begin{table*}[t!]
\begin{center}
\begin{tabular}{ |c||c|c|c| }
  \hline
  Unitary &Gate cost&Qubit cost\\
  \hline \hline
 $\e^{-i H_0 T}$ &${\cal O}(Ld)$ & ${\cal O}(1)$   \\\hline
%$\e^{-i\Delta t D_0}$ & diagonal evolution &${\cal O}(C_{D_0})$ & ${\cal O}(1)$     \\\hline
%$W$ & $W=B^{\dagger} U_C B$ &${\cal O}(Q^2 +Q M (C_{\Delta D_0}+k_{\rm od} +\log M))$& ${\cal O}(Q\log M)$    \\\hline
%$B$ & LCU state preparation &${\cal O}(QM)$ & ${\cal O}(Q\log M)$    \\\hline
$V_c$ &  $\mathcal{O}(Q^2+QM(k_{od}+\log M)+QMK(C_D+C_{\Delta H_0}+C_{\Lambda}))$ & ${\cal O}(Q\log MK)$    \\\hline
%$U_{CP}$ & controlled permutation &${\cal O}(QM(k_{\rm od}+\log M))$ & ${\cal O}(Q\log M)$    \\\hline
%$U_{C\Phi}$ & controlled phase & ${\cal O}(Q^2 + Q M (C_{\Delta D_0}+k_{\rm od} +\log M))$ & ${\cal O}(Q\log M)$   \\\hline
$U_I(T)$ & $\mathcal{O}\left(r\left(Q^2+QM(k_{od}+\log M)+QMK(C_D+C_{\Delta H_0}+C_{\Lambda})\right)\right)$ & ${\cal O}(Q\log MK)$    \\\hline
\end{tabular}
\end{center}
\caption{\label{tbl:resource}{\bf A summary of resources for the circuit.}}
\end{table*}

\subsubsection{Example advantages of the algorithm}\label{advantages}
To illustrate how our simulation algorithm can provide speedups over existing algorithms, we focus in this subsection on two types of Hamiltonian systems: highly oscillating systems and decaying systems. 

The cost of our algorithm is independent of the oscillation rates of the dynamics, whereas the cost of any simulation algorithm (e.g., \cite{Berry2020timedependent,Berry2019,low2019hamiltonian,PhysRevLett.106.170501}) that depends on $||dH(t)/dt ||$ would depend on oscillation rates of the system. To illustrate this advantage, consider a two-level system with a Hamiltonian 
\begin{equation}
H(t)=h Z+\Gamma \left(\e^{-i \alpha t }|0\rangle \langle 1|+ \e^{i\alpha t}|1\rangle\langle 0|\right)=H_0+V(t), 
\end{equation}
where $h,\Gamma,\alpha\in \mathbb{R}$, $H_0=h Z$ and $V(t)=\Gamma \left(\e^{-i \alpha t }|0\rangle \langle 1|+ \e^{i\alpha t}|1\rangle\langle 0|\right)$. In this case, we have $k_{od}=M=K=1$ and $\lambda=0$. The gate cost of simulating $U(T)$ scales as 
\begin{equation}
\mathcal{O}\left(\Gamma T\left(\frac{\log \Gamma T/\epsilon}{\log\log\Gamma T/\epsilon}\right)^2\right),
\end{equation}
which is independent of $\alpha$. This means the simulation cost remains the same even if $\alpha$ becomes arbitrarily large.  One can realize the absence of $\alpha$ owing to the fact that phases are explicitly integrated out into an integral-free expansion series, where the bound of each term does not depend on the oscillations (due to Identity 2).  Therefore, our simulation can be significantly more effective when the time dependence of the Hamiltonian has very high frequencies. Note that while the example above was given for a simple qubit system with pure oscillation, the frequency-independence in cost holds for any system.

%\subsubsection{Advantages of the algorithm II: Decaying systems}\label{advantages2}
Another class of systems for which our algorithm can provide seedup are Hamiltonians with exponential decays, i.e., $\lambda< 0$.  For concreteness, consider the Hamiltonian 
\begin{equation}
H(t)=hZ+\Gamma \e^{-\alpha t}X=H_0+V(t),
\end{equation}
where $h,\Gamma\in\mathbb{R}$ and $\alpha>0$ and $H_0=hZ$ and $V(t)=\Gamma \e^{-\alpha t}X$.  In this case, $\lambda=-\alpha$ and $||V(t)||_{\max}=\Gamma \e^{-\alpha t}$. 

The $L^1$-norm defined in  Ref.~\cite{Berry2020timedependent} is $\int^{T}_0||H(t)||_{\max}dt$, which has a linear scaling $\mathcal{O}(hT)$ with the simulation duration $T$, whereas our discretized $L^1$-norm $\sum^{r-1}_{w=0} ||V(t_w)||_{\max}\Delta t_w$ tends to a constant in the long time limit. This can be seen from the fact that the partition terminates at a large enough time $t_w$ ($\leq T$), where $\Delta t_w=T-t_w$ becomes the final simulation step, as described in Sec. \ref{timepartition}.  The above results also hold for any combination of exponential decays (even when these are multiplied by oscillatory terms) with which different time decay dependencies may be constructed.

\subsection{Hamiltonians with arbitrary time dependence}\label{generaltimedependence}
The simulation algorithm invokes a switch to the interaction picture, by dividing the Hamiltonian into a static diagonal part $H_0$ and a time-dependent Hermitian operator $V(t)$. The $V(t)$ is expanded using permutations and exponential sums as presented in Eq. (\ref{Vexpand2}).  There, we assume that the time dependence can be expressed as exponential sums with a finite number of terms, $K$. Although this assumption holds for many models  (e.g., when the time dependencies are some combinations of trigonometric functions and exponential decays), the exponential series generally requires an infinite sum (e.g., a Fourier series).  A straightforward procedure to obtain a finite sum approximation is via a truncated Fourier series. As an example, let us consider a polynomial function of time, i.e., $f(t)=\sum_{l=0}^p c_l t^l$. Using the proof of Theorem 8.14 in Ref.~\cite{Rudin}, it can be shown that a truncated Fourier series of $f(t)$ is $\mathcal{O}(\epsilon)$ close to $f(t)$ when the truncation order is $\mathcal{O}(1/\epsilon)$. We also note that, other than Fourier series, there have been numerous studies~\cite{Norvidas2010,Beylkin2005,Beylkin2010,Braess2009,Wiscombe1977} regarding finding an exponential-sum approximation of a function. Some of them, e.g.,~\cite{Beylkin2005}, provide efficient algorithms with logarithmically scaling terms (with respect to the inverse of a required accuracy). These results suggest that efficient methods for finding the exponential-sum decompositions of the time dependences of $V(t)$ can exist in many cases. 

Suppose that $V(t)$ is approximated by a finite series of exponential sum. The resulting error of the unitary evolution, due to the Hamiltonian approximation, scales only at most linearly with the evolution duration.  This can be shown using the following property.  Given two time-dependent Hamiltonians $H_1(t)$ and $H_2(t)$ such that
\begin{equation}
|| H_1(t)-H_2(t)||\leq \epsilon \ \ \text{for all }t\in[0,T], 
\end{equation}
then 
\begin{equation}\label{UBound}
||U_1(T,0)-U_2(T,0) ||\equiv \left| \left| \mathcal{T}\text{exp}\left[-i\int_0^T H_1(t)dt \right] - \mathcal{T}\text{exp}\left[-i\int_0^T H_2(t)dt \right] \right|\right|\leq \epsilon T. 
\end{equation}
This holds true for any norm $||\cdot||$. Before proving this, we first note a property of the so-called Subadditivity of error in implementing unitaries~\cite{NielsenChuang}. It says that for unitaries $U_1, U_2, V_3$ and $V_4$, we have 
\begin{equation}
||U_2U_1-V_2V_1 ||\leq ||U_2-V_2 ||+||U_1-V_1 ||.
\end{equation}
This can be easily shown by
\begin{align}
||U_2U_1-V_2V_1 ||&=||U_2U_1-V_2U_1+V_2U_1 -V_2V_1 || \nonumber \\
&\leq ||(U_2-V_2)U_1 ||+||V_2(U_1-V_1) ||\leq ||U_2-V_2 ||+||U_1-V_1 ||, 
\end{align}
where the basic operator norm inequalities are used. Now we prove the bound in Eq.~(\ref{UBound}). We divide $T$ into $n$ segments such that each segment has width $T/n$. We can rewrite the time evolution operators as 
\begin{align}
&U_{1}(T,0)=U_{1}\left(T,\frac{n-1}{n}T\right)\cdots U_{1}\left(\frac{T}{n},0\right) \nonumber\\
&U_{2}(T,0)=U_{2}\left(T,\frac{n-1}{n}T\right)\cdots U_{2}\left(\frac{T}{n},0\right) \nonumber.
\end{align}
Repeatedly using the subadditivity of error, we have
\begin{align}
&||U_1(T,0)-U_2(T,0) ||\leq \sum_{m=1}^n \left|\left|U_1\left(\frac{m T}{n}, \frac{(m-1)T}{n} \right)-U_2\left(\frac{m T}{n}, \frac{(m-1)T}{n}  \right)\right|\right| \nonumber\\
&=\sum_{m=1}^n \left|\left|   -i\int_{\frac{(m-1) T}{n}}^{\frac{mT}{n}} \left[H_1(t)-H_2(t)\right]dt+ (-i)^2\int^{\frac{mT}{n}} _{\frac{(m-1) T}{n}} dt_2\int_{\frac{(m-1) T}{n}}^{t_2} dt_1 \left[H_1(t_2)H_1(t_1)-H_2(t_2)H_2(t_1)\right]+\cdots\right|\right| \nonumber\\
&\leq\sum_{m=1}^n \epsilon \frac{T}{n}+ \sum_{m=1}^n\mathcal{O}\left[\left(\frac{T}{n}\right)^2\right]= \epsilon T + \mathcal{O}\left[\frac{T^2}{n}\right].
\end{align}
Since this inequality holds for any $n$, we can take $n\to \infty$ and it yields Eq.~(\ref{UBound}) as claimed. 

Now we apply this property to the simulation of $U_I(T)$. Suppose that we have an $\tilde{\delta}-$accurate approximation of $V(t)$, i.e., $||\tilde{V}(t)-V(t)||\leq \tilde{\delta}$ for all $t\in[0,T]$, where $\tilde{V}(t)$ is the finite exponential-sum approximation of $V(t)$. The accumulative error from this approximation is bounded by $\tilde{\delta}T$ and the overall error is $\mathcal{O}(\tilde{\delta} T+\delta)$, where $\delta$ is the error from LCU implementation. Recall that $\tilde{\delta}$ is closely related to $K$ (the number of terms). Although it is intuitive that a larger $K$ can allow for a smaller $\tilde{\delta}$, the explicit relation between the two largely depends on the model and the expansion method. Nonetheless, we can expect $K$ to scale at least linearly with $1/\tilde{\delta}$ for many cases, e.g., aforementioned truncated Fourier series for a polynomial. 

The simulation cost also depends on $M$, the number of terms in the permutation expansion. This quantity usually scales linearly with the system size and can be easily determined. For example, a typical spin model usually involves a sum of tensor products of Pauli-$X$'s (or $Y$'s) and Pauli-$Z$'s. Each tensor product represents an interaction between qubits on certain lattice sites. Due to the common locality constraint that prevents a qubit interacting with the ones arbitrarily far apart, the number of interacting terms, $M$, scales at most linearly with the number of qubits. In addition, a tensor product of Pauli operators can be easily separated into a product of diagonal matrix and a permutation, e.g., $X\otimes X\otimes Y= (I\otimes I\otimes -iZ)(X\otimes X\otimes X)$. We conclude that $M$ will have modest linear scaling for most practical models.

\section{Alternative scheme and reduction to the time-independent case}\label{TIndep}
In this section, we provide an alternative yet equivalent scheme for the dynamical simulation, one that will allow us to establish an immediate connection to the time-independent Hamiltonian simulation formalism (specifically to the scheme presented in Ref.~\cite{kalev2020simulating}), in which $H(t)$ is assumed constant in time.  

In previous sections, we have chosen to partition the interaction-picture unitary $U_I(T)$ into short time segments and then follow its execution by the application of a diagonal $\e^{i H_0 T}$ bringing it back to the Schr{\"o}dinger picture. Here, we show that the Schr{\"o}dinger picture $U(T)$ can be partitioned similarly. 

Recalling the expansion of $U_I(t_w+\Delta t_w, t_w)$ in Eq.~(\ref{U_t_dt}), we have
\begin{align}
U_I(t_w+\Delta t_w, t_w)|z\rangle =\sum_{q=0}^{\infty}\sum_{\bold{i}_q}\sum_{\bold{k}_q}(-i)^q \e^{t_w x_1} \e^{\Delta t_w[x_1,x_2,\cdots, x_q,0]} d^{(\bold{k}_q)}_{\bold{i}_q,z} P_{\bold{i}_q}|z\rangle, \nonumber
\end{align}
with 
\begin{equation}
x_1=i\left(E_{z_{\bold{i}_q}}-E_{z}\right)+\sum_{l=1}^q \lambda^{(k_l)}_{i_l,z_{\bold{i}_l}}. \nonumber
\end{equation}
Breaking the $\e^{t_w x_1} $ phase, we get:
\begin{align}
&U_I(t_w+\Delta t_w, t_w)|z\rangle \nonumber \\
&=\sum_{q=0}^{\infty}\sum_{\bold{i}_q}\sum_{\bold{k}_q}(-i)^q 
\e^{t_w \sum_{l=1}^q \lambda^{(k_l)}_{i_l,z_{\bold{i}_l}}}
\e^{-i t_w E_z} 
\e^{i(t_w+\Delta t_w) E_{z_{\bold{i}_q}}} 
\e^{-i \Delta t_w  E_{z_{\bold{i}_q}}} 
\e^{\Delta t_w[x_1,x_2,\cdots, x_q,0]} d^{(\bold{k}_q)}_{\bold{i}_q,z} P_{\bold{i}_q}|z\rangle. 
\end{align}
We find that 
\begin{equation}
\e^{-i H_0 (t_w+\Delta t_w)} U_I(t_w+\Delta t_w, t_w) =\widetilde{U}_I(t_w+\Delta t_w, t_w) \e^{-i H_0 t_w}
\end{equation}
where
\begin{equation}
\widetilde{U}_I(t_w+\Delta t_w, t_w)=
\sum_z \sum_{q=0}^{\infty}\sum_{\bold{i}_q}\sum_{\bold{k}_q}(-i)^q 
\e^{t_w \sum_{l=1}^q \lambda^{(k_l)}_{i_l,z_{\bold{i}_l}}}
\e^{-i \Delta t_w  E_{z_{\bold{i}_q}}} 
\e^{\Delta t_w[x_1,x_2,\cdots, x_q,0]} d^{(\bold{k}_q)}_{\bold{i}_q,z} P_{\bold{i}_q}|z\rangle \langle z| \,.
\end{equation}
Inspecting the full unitary evolution, we observe 
\begin{align}
U(T)&= \e^{-i H_0 T}U_I(T) \nonumber\\
&= \e^{-i H_0 T} U_I(T, t_{r-1})U_I(t_{r-1},t_{r-2})\cdots U_I(t_1,0)=\widetilde{U}_I(T, t_{r-1})D(t_{r-1})U_I(t_{r-1},t_{r-2})\cdots U_I(t_1,0).
\end{align}
The evolution operator $U(T)$ can be simplifies as 
\begin{equation}
U(T)= \widetilde{U}_I(T, t_{r-1})\widetilde{U}_I(t_{r-1},t_{r-2})\cdots \widetilde{U}_I(t_1,0) \,, 
\end{equation}
eliminating the diagonal piece. Each $\widetilde{U}_I(t_w+\Delta t_w, t_w)$ can be rewritten as:
\begin{align}
&\widetilde{U}_I(t_w+\Delta t_w, t_w) \nonumber\\
&=\sum_z \sum_{q=0}^{\infty}\sum_{\bold{i}_q}\sum_{\bold{k}_q}(-i)^q 
\e^{(t_w+\Delta t_w) \sum_{l=1}^q \lambda^{(k_l)}_{i_l,z_{\bold{i}_l}}}
\e^{-\Delta t_w  (i E_{z_{\bold{i}_q}}+\sum_{l=1}^q \lambda^{(k_l)}_{i_l,z_{\bold{i}_l}})} 
\e^{\Delta t_w[x_1,x_2,\cdots, x_q,0]} d^{(\bold{k}_q)}_{\bold{i}_q,z} P_{\bold{i}_q}|z\rangle \langle z| \,.
\end{align}
The factor $\e^{-\Delta t_w  (i E_{z_{\bold{i}_q}}+\sum_{l=1}^q \lambda^{(k_l)}_{i_l,z_{\bold{i}_l}})} $ can be absorbed into the divided difference:
\begin{equation}
\widetilde{U}_I(t_w+\Delta t_w, t_w)=
\sum_z \sum_{q=0}^{\infty}\sum_{\bold{i}_q}\sum_{\bold{k}_q}(-i)^q 
\e^{(t_w+\Delta t_w) \sum_{l=1}^q \lambda^{(k_l)}_{i_l,z_{\bold{i}_l}}}
\e^{\Delta t_w[\tilde{y}_1,\tilde{y}_2,\cdots, \tilde{y}_q,\tilde{y}_{q+1}]} d^{(\bold{k}_q)}_{\bold{i}_q,z} P_{\bold{i}_q}|z\rangle \langle z| \,.
\end{equation}
with
\begin{equation}
\tilde{y}_j = x_j -\left(i E_{z_{\bold{i}_q}}+\sum_{l=1}^q \lambda^{(k_l)}_{i_l,z_{\bold{i}_l}}\right) 
=i\left(E_{z_{\bold{i}_q}}-E_{z_{\bold{i}_{j-1}}}\right)+\sum_{l=j}^q \lambda^{(k_l)}_{i_l,z_{\bold{i}_l}}
-\left(i E_{z_{\bold{i}_q}}+\sum_{l=1}^q \lambda^{(k_l)}_{i_l,z_{\bold{i}_l}}\right) \,,
\end{equation}
which simplifies to
\begin{equation}
\tilde{y}_j =-i E_{z_{\bold{i}_{j-1}}}-\sum_{l=1}^{j-1} \lambda^{(k_l)}_{i_l,z_{\bold{i}_l}} \,.
\end{equation}
By inserting additional $i \Delta t_w E_z$ phases into the divided differences, we can rewrite
\begin{equation}
\widetilde{U}_I(t_w+\Delta t_w, t_w) =\left(
\sum_z \sum_{q=0}^{\infty}\sum_{\bold{i}_q}\sum_{\bold{k}_q}(-i)^q 
\e^{(t_w+\Delta t_w) \sum_{l=1}^q \lambda^{(k_l)}_{i_l,z_{\bold{i}_l}}}
\e^{\Delta t_w[y_1,y_2,\cdots, y_q,y_{q+1}]} d^{(\bold{k}_q)}_{\bold{i}_q,z} P_{\bold{i}_q}|z\rangle \langle z|  \right) \e^{-i H_0 \Delta t_w} \,.
\end{equation}
with
\begin{equation}
y_j =\tilde{y}_j + i E_z = -i (E_{z_{\bold{i}_{j-1}}} - E_z) -\sum_{l=1}^{j-1} \lambda^{(k_l)}_{i_l,z_{\bold{i}_l}}
=-i \Delta E_{z_{\bold{i}_{j-1}}}-\sum_{l=1}^{j-1} \lambda^{(k_l)}_{i_l,z_{\bold{i}_l}}\,.
\end{equation}
Now, we can write $U(T)$ as alternating off-diagonal and diagonal unitaries:
\begin{equation}
U(T) =\prod_w\widetilde{U}_I(t_w+\Delta t_w, t_w)\equiv\prod_w  U_{\textrm{od}}(t_w+\Delta t_w, t_w) \e^{-i H_0 \Delta t_w}\,.
\end{equation}
When $H(t)$ becomes time-independent, $ \lambda^{(k_l)}_{i_l,z_{\bold{i}_l}}=0$ and $\Delta t_w=\Delta t=\ln 2/\Gamma $. To synchronize the notation with \cite{kalev2020simulating}, we identify $H_0=D_0$ and $U_{\textrm{od}}(t_w+\Delta t_w, t_w)=U_{\textrm{od}}$. The evolution operator becomes $U(T)=U_{\textrm{od}}\e^{-iD_0 \Delta t}\cdots U_{\textrm{od}}\e^{-iD_0 \Delta t}$, which coincides with \cite{kalev2020simulating}.

\section{Conclusions}\label{sec:conc}
We presented a quantum algorithm for simulating the evolution operator generated from a time-dependent Hamiltonian. The algorithm involves a permutation expansion for the interaction Hamiltonian, a switch to the interaction-picture,  and the incorporation of the LCU technique. Combining the permutation expansion with the Dyson series has led to an integral-free representation for the interaction-picture unitary with coefficients involving the notion of divided differences with complex inputs. 

%In our construction of the algorithm, we implicitly expand the Hamiltonian as a finite exponential sum. Although this can be done naturally in many Hamiltonians, there are cases when the decomposition requires an infinite sum to be exact.  However, we mention that there are various algorithms for efficiently finding an exponential sum approximation of a function. Some of them require a fewer terms comparing to a truncated Fourier series, a straightforward way to obtain a  finite series approximation. 

 We found that our expansion  allowed us to adjust the time steps based on the dynamical characteristics of the Hamiltonian, providing a resource saving as compared to the equal-size partition with the largest bound. This further resulted in a gate resource that scales  with an $L^1$-norm-like scaling with respect only to the `non-static' norm of the Hamiltonian. 

Specifically, we demonstrated that for systems with a decaying non-static component, the resources do not scale with the total evolution time asymptotically. Furthermore, the simulation cost is independent of the frequencies, implying a significant advantage for systems with highly oscillating components.

\begin{acknowledgments}
This work is supported by the U.S. Department of Energy (DOE), Office of Science, Basic Energy Sciences (BES) under Award No. DE-SC0020280.
\end{acknowledgments}
\bibliography{rf}

\appendix
\section*{Appendix}
\section{Properties of divided difference}\label{DDapendix}
We begin with a formal definition of divided difference for complex-valued functions and follow with some properties that will be of use to us when deriving the new bound. The main results are derived for the exponential functions.
\begin{define}\label{def1}
Let $\mathbb{U}$ be an open subset of $\mathbb{C}$, and $f:\mathbb{U}\to\mathbb{C}$ is analytic in $\mathbb{U}$. For any non-negative integer $q$ and $x_0, x_1,\cdots,x_q\in \mathbb{U}$, the divided difference of $f$ is denoted as $f[x_0, x_1,\cdots,x_q]$. If $q=0$, $f[x_0]\equiv f(x_0)$. Suppose $\{x_0, x_1,\cdots,x_q\}$ has $r$ distinct elements. Let $S=\{x_{\sigma(0)},x_{\sigma(1)},\cdots,x_{\sigma(q)}\}$ be a sorted set of $\{x_0, x_1,\cdots,x_q\}$, i.e., there exists a permutation $\sigma$ such that the first $n_1$ elements of $S$ are equal and the following $n_2$ elements of $S$ are equal and so on and so forth. There are $r$ same-element clusters and $\sum_{i=1}^{r}n_i=q+1$. The divided difference of $f$ is defined as
$$
f[x_0, x_1,\cdots,x_q]=
\begin{cases}
\frac{f[x_{\sigma(1)},\cdots,x_{\sigma(q)}]-f[x_{\sigma(0)},\cdots,x_{\sigma(q-1)}]}{x_{\sigma(q)}-x_{\sigma(0)}} & \text{if } r>1, \\
\frac{f^{(q)}(x_0)}{q!} & \text{if } r=1,
\end{cases}
$$
where $f^{(q)}$ denotes the $q$th derivative of $f$.
\end{define}
Although the above sorting procedure is not unique, it can be shown that any choice of the permutation gives the same result, and hence the definition is well-defined. 

The divided difference involves a recursive relation that connects a $q+1$ input case to two $q$ cases. For $q=1$,
\begin{align}
f[x_0, x_1]=
\begin{cases}
\frac{f(x_1)-f(x_0)}{x_1-x_0} & \text{if } x_0\neq x_1, \\
 f'(x_0) & \text{if } x_0=x_1.
\end{cases}
\end{align}
For $q=2$, and suppose $x_0$, $x_1$ and $x_2$ are all distinct,
\begin{align}
f[x_0,x_1,x_2]&=\frac{\frac{f(x_2)-f(x_1)}{x_2-x_1}-\frac{f(x_1)-f(x_0)}{x_1-x_0}}{x_2-x_0} \nonumber\\
&=\frac{f(x_0)}{(x_0-x_1)(x_0-x_2)}+\frac{f(x_1)}{(x_1-x_2)(x_1-x_0)}+\frac{f(x_2)}{(x_2-x_0)(x_2-x_1)}. 
\end{align}
In fact, it can be shown that for distinct $x_0,x_1,\cdots,x_q$, 
\begin{align}
f[x_0,x_1,\cdots,x_q]=\sum^{q}_{i=0}\frac{f(x_i)}{\prod_{k\neq i}(x_i-x_k)}.
\end{align}

\begin{rmk}
Since any analytic function admits a Taylor expansion representation and the divided difference is a linear functional, the divided difference of an analytic function $f$ has a series expansion form, i.e., for $x_0,\cdots,x_q$ and $y$ in $f$'s analytic domain,
\begin{equation}
f[x_0,\cdots,x_q]=\sum_{n=0}^{\infty}\frac{f^{(n)}(y)}{n!}p_{n|y}[x_0,\cdots,x_q],
\end{equation}
where $p_{n|y}(x)\equiv (x-y)^n$. Because $p_{n|y}[x_0,\cdots,x_q]=0$ for all $n<q$, the non-vanishing term of the series starts from the $q$th order.
\end{rmk}

For simplicity, we denote the divided difference for the exponential function as $\e^{[x_0,\cdots,x_q]}$, i.e.,
\begin{equation}
\e^{[x_0,\cdots,x_q]}\equiv f[x_0,\cdots,x_q], \ \text{where} \ f(x)=\e^x.
\end{equation}

\begin{prpty}\label{pty0}
For any non-negative integer $q$ and $x_0, x_1,\cdots,x_q\in \mathbb{C}$,
\begin{equation}
\e^{[x_0,x_1,\cdots,x_q]}=\e^{x_0}\e^{[0,x_1-x_0,\cdots,x_q-x_0]}.
\end{equation}
\end{prpty}
This property and the fact that divided differences are permutation symmetric among inputs imply that any input can be factored out of the divided difference by subtracting it from every entry.

\begin{prpty}\label{pty2}
For any non-negative integer $q$ and $x_0, x_1,\cdots,x_q\in \mathbb{C}$,
\begin{equation}
\e^{[x_0,x_1,\cdots,x_q]}=\displaystyle\sum^{\infty}_{n=q}\frac{1}{n!}\sum_{\sum k_j=n-q}\prod^{q}_{j=0}(x_j)^{k_j}. 
\end{equation}
\end{prpty}
An equivalent definition of divided difference for an analytic function is via its Taylor expansion. It amounts to apply divided difference on every order of the series. Since any polynomial of order less than $q$ is annihilated, the series starts from the order $q$. $Property$ \ref{pty2} is derived from the Taylor expansion of $\e^x$ with respect to the origin. 

\begin{lemma}\label{lemma1}
For any non-negative integer $q$ and $x_0, x_1,\cdots,x_q\in \mathbb{C}$,
\begin{equation}
\int_0^1 a^q {\rm e}^{[ax_0,ax_1,\cdots,ax_q]}da={\rm e}^{[0,x_0,x_1,\cdots,x_q]}.
\end{equation}
\end{lemma}

\begin{proof}
This can be observed from the series expansion of the divided difference for the exponential function, i.e., from $Property$ \ref{pty2},
\begin{align}
a^q \e^{[ax_0,ax_1,\cdots,ax_q]} &=a^q\displaystyle\sum^{\infty}_{n=q}\frac{1}{n!}\sum_{\sum k_j=n-q}\prod^{q}_{j=0}(ax_j)^{k_j} \nonumber\\
&=a^q\displaystyle\sum^{\infty}_{n=q}\frac{1}{n!}\sum_{\sum k_j=n-q}a^{n-q}\prod^{q}_{j=0}(x_j)^{k_j}=\displaystyle\sum^{\infty}_{n=q}\frac{a^n}{n!}\sum_{\sum k_j=n-q}\prod^{q}_{j=0}(x_j)^{k_j}. 
\end{align}
Performing term-by-term integration over $a$ on both side, we have 
\begin{align}
\int_0^1 a^q \e^{[ax_0,ax_1,\cdots,ax_q]}da&= \displaystyle\sum^{\infty}_{n=q}\left(\int_0^1\frac{a^n}{n!}da\right) \sum_{\sum k_j=n-q}\prod^{q}_{j=0}(x_j)^{k_j} \nonumber\\
&=\displaystyle\sum^{\infty}_{n=q}\frac{1}{(n+1)!}\sum_{\sum k_j=n-q}\prod^{q}_{j=0}(x_j)^{k_j}= \e^{[0,x_0,x_1,\cdots,x_q]},
\end{align} 
where the last equality follows from the series expansion representation of $\e^{[0,x_0,x_1,\cdots,x_q]}$. This completes the proof.
\end{proof}

\begin{cor}\label{cor1}
Let $f(x)={\rm e}^{tx}$, where $t\in\mathbb{R}$ and $x\in\mathbb{C}$. We denote ${\rm e}^{t[x_0,\cdots,x_q]}\equiv f[x_0,\cdots,x_q]$, where $x_0,\cdots,x_q\in\mathbb{C}$. For any $\tau\in\mathbb{R}$,
\begin{equation}
\int^{\tau}_0 {\rm e}^{t[x_0,\cdots,x_q]}dt={\rm e}^{\tau[0,x_0,\cdots,x_q]}.
\end{equation}
\end{cor}
This can be verified by evaluating the series expansion form on both side, by a similar manner in the proof of Lemma \ref{lemma1} .

With these properties, we are ready to prove the bound in Identity~\ref{bound1} in the main context. 
\setcounter{thm}{0}
\begin{thm}\label{appendixthm1}
For any non-negative integer $q$ and $x_0, x_1,\cdots,x_q\in \mathbb{C}$,
\begin{equation}
\left| \e^{[x_0,x_1,\cdots,x_q]}\right|\leq \e^{[\Re(x_0),\Re(x_1),\cdots,\Re(x_q)]}, \label{bound}
\end{equation}
where $\Re(\cdot)$ gives the real part of the input.
\end{thm}
\begin{proof}
We proceed by induction. Eq. (\ref{bound}) is trivially satisfied with the equality when $q=0$. For the case $q=1$, we have 
\begin{align}
\left| \e^{[x_0,x_1]}\right|&=\left| \e^{x_0}\right|\left| \e^{[0,x_1-x_0]}\right| \nonumber\\
&=\e^{\Re(x_0)} \left|\int_0^1 a \e^{a(x_1-x_0)}da\right| \nonumber\\
&\leq \e^{\Re(x_0)}\int_0^1 a \left| \e^{a(x_1-x_0)} \right|da \nonumber\\
&=\e^{\Re(x_0)}\int_0^1 a \e^{a\Re(x_1-x_0)} da \nonumber\\
&=\e^{[\Re(x_0),\Re(x_1)]},
\end{align}
where Lemma \ref{lemma1} is used. Assume that we have 
\begin{equation}
\left| \e^{[x_0,\cdots,x_q]}\right|\leq \e^{[\Re(x_0),\cdots,\Re(x_q)]}, \label{assmp1}
\end{equation}
which it is true for $q=0,1$. It follows that
\begin{align}
\left| \e^{[x_0,\cdots,x_q,x_{q+1}]}\right|&=\left| \e^{x_{q+1}} \right| \left| \e^{[0, x_0-x_{q+1},\cdots,x_q-x_{q+1}]}\right| \nonumber\\
&=\e^{\Re(x_{q+1})} \left| \int_0^1 a^q \e^{[a(x_0-x_{q+1}),\cdots,a(x_q-x_{q+1})]}da\right| \nonumber\\
&\leq \e^{\Re(x_{q+1})}  \int_0^1 a^q \left| \e^{[a(x_0-x_{q+1}),\cdots,a(x_q-x_{q+1})]}\right| da \nonumber\\
&\leq \e^{\Re(x_{q+1})} \int_0^1 a^q \e^{[\Re(ax_0-ax_{q+1}),\cdots,\Re(ax_q-ax_{q+1})]} da \nonumber\\
&=\e^{\Re(x_{q+1})} \e^{[0,\Re(x_0-x_{q+1}),\cdots,\Re(x_q-x_{q+1})]} \nonumber\\
&=\e^{[\Re(x_0),\cdots,\Re(x_q),\Re(x_{q+1})]},
\end{align}
where the second and the third equalities use Lemma \ref{lemma1} and the second inequality uses (\ref{assmp1}). This proves that the inequality holds for any number of complex inputs. 
\end{proof}

\section{Bounding $\left|\e^{\Delta t_w[x_1,x_2,\cdots, x_q,0]}\right|$}\label{forbounds}

For $\left|\e^{\Delta t_w[x_1,x_2,\cdots, x_q,0]}\right|$, we use the following theorem,  
\setcounter{thm}{0}
\begin{thm}
For any $q+1$ complex values $x_0,\cdots,x_q\in \mathbb{C}$, 
\begin{equation}
\left| {\rm e}^{[x_0,\cdots,x_q]}\right|\leq {\rm e}^{[\Re(x_0),\cdots,\Re(x_q)]}=\frac{{\rm e}^{\xi}}{q!},
\end{equation}
where $\Re(\cdot)$ denotes the real part of an input and $\xi\in \left[\min\{\Re(x_0),\cdots,\Re(x_q)\},\max\{\Re(x_0),\cdots,\Re(x_q)\}\right]$.
\end{thm}
This is proved in Appendix \ref{DDapendix}. From this, we have
\begin{equation}
 \left|\e^{\Delta t_w[x_1,\cdots, x_q,0]}\right|=({\Delta t_w})^q\left|\e^{[\Delta t_w x_1,\cdots,\Delta t_w x_q,0]}\right|\leq ({\Delta t_w})^q \e^{[\Delta t_w \Re(x_1),\cdots, \Delta t_w \Re(x_q),0]}. \label{tbound}
\end{equation} 
From the definition of $x_j$, we have
\begin{equation}
\forall j\in\{1,\cdots,q\},\ \ \ \ \  \Re(x_j)=\sum_{l=j}^q \Re\left(\lambda^{(k_l)}_{i_l,z_{\bold{i}_l}} \right)\leq (q-j+1)\lambda. 
\end{equation}
Based on the property that increasing any input in $\e^{[\cdot,\cdots,\cdot]}$ will only increase its value (can be proved by taking derivatives in the Hermite-Genocchi form), we have 
\begin{equation}
\left|\e^{\Delta t_w[x_1,\cdots, x_q,0]}\right|\leq ({\Delta t_w})^q \e^{[\Delta t_w \Re(x_1),\cdots, \Delta t_w \Re(x_q),0]}\leq ({\Delta t_w})^q \e^{[\Delta t_w q\lambda,\Delta t_w (q-1)\lambda,\cdots, \Delta t_w \lambda,0]}.
\end{equation}
Using the permutation symmetric property and Property \ref{pty0}, we have 
\begin{align}
({\Delta t_w})^q\e^{[\Delta t_w q\lambda,\Delta t_w (q-1)\lambda,\cdots, \Delta t_w \lambda,0]}&={\Delta t_w}^q\frac{\e^{[\Delta t_w q\lambda,\Delta t_w (q-1)\lambda,\cdots, \Delta t_w \lambda]}-\e^{[\Delta t_w (q-1)\lambda,\cdots, \Delta t_w \lambda,0]}}{\Delta t_w \lambda q} \nonumber\\
&=({\Delta t_w})^q\frac{\e^{\Delta t_w \lambda}-1}{\Delta t_w \lambda q}\e^{[\Delta t_w (q-1)\lambda,\cdots, \Delta t_w \lambda,0]} =\cdots= \left(\frac{\e^{\lambda \Delta t_w}-1}{\lambda}\right)^q\frac{1}{q!}.
\end{align}
Therefore, we have
\begin{equation}
\left|\e^{\Delta t_w[x_1,x_2,\cdots, x_q,0]}\right|\leq \frac{1}{q!} \left(\frac{\e^{\lambda \Delta t_w}-1}{\lambda}\right)^q.
\end{equation}

\section{LCU method review}\label{LCUrecap}
We give a brief introduction to the LCU method in this section, and we adapt the original paper~\cite{Berry2015}'s notations for a more convenient reference to readers. Suppose we have a unitary $U$, which is an infinite sum of unitaries, i.e.,
\begin{equation}
U=\sum_{j=0}^{\infty}\beta_jV_j, \label{U_expand}
\end{equation}
where $\beta_j>0$ and $V_j$ are some unitaries. A truncated series, up to order $m-1$, yields an operator 
\begin{equation}
\tilde{U}=\sum_{j=0}^{m-1}\beta_jV_j, 
\end{equation}
which approaches $U$ as $m$ increases. We perform the following procedure to effectively implement $\tilde{U}$ on a state $|\psi\rangle$ embedded in a larger system. Prepare an $m$-dimensional ancilla $|0\rangle$ and implement a unitary $B$ such that
\begin{equation}
B|0\rangle=\frac{1}{\sqrt{s}}\sum_{j=0}^{m-1}\sqrt{\beta_j}|j\rangle,
\end{equation}
where $s=\sum_{j-0}^{m-1}\beta_j$. Suppose we have access to a control unitary $V_c$ such that for each $j$,
\begin{equation}
V_c|j\rangle|\psi\rangle=|j\rangle V_j|\psi\rangle.
\end{equation}
Consider the following combination of the above operations
\begin{equation}
W\equiv\left(B^{\dagger}\otimes I\right)V_c\left( B\otimes I\right).
\end{equation}
We have
\begin{equation}
W|0\rangle|\psi\rangle=\frac{1}{s}|0\rangle \tilde{U}|\psi\rangle+\sqrt{1-\frac{1}{s^2}}|\Phi\rangle,
\end{equation}
where $|\Phi\rangle$'s ancillary part is orthogonal to $|0\rangle\langle0|$. Let us denote $P\equiv|0\rangle\langle0|\otimes I$ the orthogonal projection onto that subspace and $R\equiv I-2P$ the reflection operator with respect to $P$. It is shown that the sequence of operations $A\equiv -WRW^{\dagger}RW$,
acting on the total system is $A|0\rangle|\psi\rangle=|0\rangle \tilde{U}|\psi\rangle$ when $\tilde{U}$ is unitary and $s=2$. This procedure is the so-called Oblivious Amplitude Amplification (OAA). However, $\tilde{U}$ is in general not unitary because it is a truncated series of $U$. This nonunitarity can be accounted for when $\tilde{U}\approx U$ and $s\approx2$. More specifically, it is shown that if $||U-\tilde{U}||=\mathcal{O}(\delta)$ and $|s-2|=\mathcal{O}(\delta)$, then 
\begin{equation}
\left|\left| PA|0\rangle|\psi\rangle-|0\rangle U|\psi\rangle\right|\right|=\mathcal{O}(\delta).
\end{equation}
This means when $\tilde{U}$ is $\delta$-close to $U$ and $s$ is $\delta$-close to 2, the effect of the operator $A$ on the whole system is $\delta$-close to only $U$ acting on $|\psi\rangle$. 

Note that the condition $||U-\tilde{U}||=\mathcal{O}(\delta)$ can be satisfied when the truncation order $m$ is high enough. However, the condition $|s-2|=\mathcal{O}(\delta)$ is satisfied only when $\beta_j$ are specifically chosen. By construction, we require $s=\sum_{j=0}^{m-1}\beta_j$. If we choose $\beta_j=(\text{ln}2)^j/j!$, then 
\begin{equation}
s=\sum_{j=0}^{m-1}\frac{(\text{ln}2)^j}{j!}
\end{equation}
becomes a truncated Taylor expansion of 2, i.e., $2=\e^{\text{ln}2}$. In fact, it can be shown that the required truncation order $m$ such that $|s-2|=\mathcal{O}(\delta)$ scales like $\log(1/\delta)/\log(\log(1/\delta))$. With this $m$, it also guarantees that $||U-\tilde{U}||=\mathcal{O}(\delta)$, because  
\begin{equation}
\left|\left|U-\tilde{U}\right|\right|=\left|\left|\sum_{j=m}^{\infty}\frac{(\text{ln}2)^j}{j!}V_j\right|\right|\leq \sum_{j=m}^{\infty}\frac{(\text{ln}2)^j}{j!}=|2-s|.
\end{equation}

In summary, performing $A$ on an extended system $|0\rangle|\psi\rangle$, with $\beta_j=(\text{ln}2)^j/j!$ and $m=\mathcal{O}(\log(1/\delta)/\log\log(1/\delta))$, effectively performs $U$ on $|\psi\rangle$ with $\mathcal{O}(\delta)$ accuracy.

\section{An alternative approach for the LCU setup}\label{AltLCU}
We provide an alternative procedure for the LCU routine that leads to an exponential saving for the state preparation. Let us define 
\begin{equation}
\Gamma\equiv \max_{\forall k,i}||D^{(k)}_i||_{\max}.
\end{equation}
Re-evaluate the coefficients in Eq.~(\ref{U_t_dt}) using the $\Gamma$ above, we have
\begin{align}
\left| \e^{t_w x_1} \e^{\Delta t_w[x_1,x_2,\cdots, x_q,0]} d^{(\bold{k}_q)}_{\bold{i}_q,z}\right|=\frac{\left(\Gamma\widetilde{\Delta t}_w \e^{ t_w \lambda}\right)^q}{q!}\cos\left[\phi^{(\bold{k}_q)}_{\bold{i}_q,z}\right] \e^{i \theta^{(\bold{k}_q)}_{\bold{i}_q,z}} .
\end{align}
The evolution operator from $t_w$ to $t_w+\Delta t_w$ becomes
\begin{align}
&U_I(t_w+\Delta t_w, t_w)=\sum_z U_I(t_w+\Delta t_w, t_w)|z\rangle\langle z|  \nonumber\\
&=\sum_z\sum_{q=0}^{\infty}\sum_{\bold{i}_q}\sum_{\bold{k}_q}(-i)^q \frac{\left(\Gamma\widetilde{\Delta t}_w \e^{ t_w \lambda}  \right)^q }{2q!}\left( \e^{i \phi^{(\bold{k}_q)}_{\bold{i}_q,z}+i \theta^{(\bold{k}_q)}_{\bold{i}_q,z}}+\e^{-i \phi^{(\bold{k}_q)}_{\bold{i}_q,z}+i \theta^{(\bold{k}_q)}_{\bold{i}_q,z}}\right) P_{\bold{i}_q} |z\rangle\langle z| \nonumber\\
&=\sum_{q=0}^{\infty}\frac{\left(\Gamma\widetilde{\Delta t}_w \e^{ t_w \lambda}\right)^q }{2q!}\sum_{\bold{i}_q}\sum_{\bold{k}_q}\sum_{x=\pm}(-i)^q  P_{\bold{i}_q} \Phi^{(\bold{k}_q,w)}_{\bold{i}_q,x}. 
\end{align}
The required state $|\psi_0\rangle$ for LCU becomes
\begin{equation}
|\psi_0\rangle=\frac{1}{\sqrt{s} } \displaystyle\sum_{q=0}^{Q}\sqrt{\frac{\left(\Gamma\widetilde{\Delta t}_w \e^{ t_w \lambda}  \right)^q}{2q!}}\sum_{\bold{i}_q}\sum_{\bold{k}_q}\sum_{x=0,1}|\bold{i}_q\rangle|\bold{k}_q\rangle|x\rangle, \label{newpsi0}
\end{equation}
where $s$ is the normalization factor, i.e.,
\begin{equation}
s=\displaystyle\sum_{q=0}^{Q}\frac{\left(MK\Gamma\widetilde{\Delta t}_w \e^{ t_w \lambda}  \right)^q}{q!}.
\end{equation}
To prepare the state (\ref{newpsi0}), we first prepare a state in the following form,
\begin{equation}
\frac{1}{\sqrt{s} } \displaystyle\sum_{q=0}^{Q}\sqrt{\frac{\left(MK\Gamma\widetilde{\Delta t}_w \e^{ t_w \lambda} \right)^q}{q!}} |1\rangle^{\otimes q}|0\rangle^{\otimes (Q-q)}|1\rangle^{\otimes q}|0\rangle^{\otimes (Q-q)}. \label{newmidstate}
\end{equation}
Subsequently, for each $|1\rangle$ in the first $Q$ registers ($\bold{i}_q$ part), we transform it to $(1/\sqrt{M})\sum_{i=0}^{M}|i\rangle$, and for each $|1\rangle$ in the later $Q$ registers ($\bold{k}_q$ part), we transform it to $(1/\sqrt{K})\sum_{k=1}^K|k\rangle$. The state (\ref{newmidstate}) becomes
\begin{equation}
\frac{1}{\sqrt{s}} \displaystyle\sum_{q=0}^{Q}\sqrt{\frac{\left(MK\Gamma\widetilde{\Delta t}_w \e^{ t_w \lambda} \right)^q}{q!}}\sum_{\bold{i}_q} \frac{1}{\sqrt{M^q}}|\bold{i}_q\rangle\sum_{\bold{k}_q}\frac{1}{\sqrt{K^q}}|\bold{k}_q\rangle=\frac{1}{\sqrt{s} } \displaystyle\sum_{q=0}^{Q}\sum_{\bold{i}_q}\sum_{\bold{k}_q}\sqrt{\frac{\left(\Gamma\widetilde{\Delta t}_w \e^{ t_w \lambda}  \right)^q}{q!}}|\bold{i}_q\rangle|\bold{k}_q\rangle, 
\end{equation}
which is the required  $|\psi_0\rangle$ in (\ref{newpsi0}), when combined with $|x\rangle$. Note that since the transformations $|1\rangle\to(1/\sqrt{M})\sum_{i=0}^{M}|i\rangle$ and $|1\rangle\to(1/\sqrt{K})\sum_{k=1}^K|k\rangle$ are mappings to the equally distributed state, they can be done with a column of parallel Hadamard gates, which has a gate cost $\mathcal{O}(\log(MK))$. This provides an exponential saving comparing to $\mathcal{O}(MK)$ given in the main context. This saving can be apparent when $MK$ becomes large. However, this can create an overhead in the required number of repetitions. Indeed, we have $\Gamma \e^{t_w\lambda }=MK\max_{\forall k,i}|| D^{(k)}_i ||_{\max}\e^{t_w\lambda }$ here comparing to $\Gamma(t_w)=\sum_{i}\sum_k|| D^{(k)}_i ||_{\max}\e^{t_w\lambda_{(i,k)}}$ in the main context, and the overall simulation cost monotonically increases with this quantity. If only a few $||D^{(k)}_i||_{\max}\e^{t_w\lambda_{(i,k)}}$ is much larger than the others such that $MK\max_{\forall k,i}|| D^{(k)}_i ||_{\max}\e^{t_w\lambda}\gg\sum_{i}\sum_k|| D^{(k)}_i ||_{\max}\e^{t_w\lambda_{(i,k)}}$, then the method provided in the main context is preferred. Depending on the models, one may favors one over the other.

\end{document}